\documentclass[12pt]{amsart}
\usepackage{graphicx}
\usepackage{color}
\usepackage{amsmath}
\usepackage{amssymb}
\usepackage{amscd}
\usepackage{amsthm}
\usepackage{amsopn}
\usepackage{xspace}
\usepackage{verbatim}
\usepackage{amsmath}
\usepackage{amsfonts}
\usepackage{epic}
\usepackage{eepic}
\usepackage{stmaryrd}
\usepackage{empheq}
\usepackage[active]{srcltx}
\usepackage[a4paper,margin=3cm]{geometry}
\definecolor{red}{rgb}{1,0,0}
\definecolor{green}{rgb}{0,1,0}
\definecolor{SeaGreen}{RGB}{46,139,87}
\definecolor{Maroon}{RGB}{128,0,0}

\newcommand{\N}{\mathbb{N}}

\newcommand{\C}{{\mathbb{C}}}
\newcommand{\R}{{\mathbb{R}}}

\newcommand{\A}{\mathcal A}
\newcommand{\B}{\mathcal B}

\def\Dg {{\mathcal D}}

\newcommand{\LL}{\mathcal L}

\newcommand{\OO}{\mathcal O}

\def\Sg {{\mathcal S}}

\def\Vg {{\mathcal V}}

\newcommand{\sca}[2]{\langle#1,#2\rangle}

\def\sign{\text{\rm sign\,}}

\newcommand{\supp}{\rm Supp\,}

\renewcommand {\Re}{{\rm Re\,}}
\renewcommand{\Im}{{\rm Im\,}}
\newcommand {\pa}{\partial}

\def\0{\mathbf  0}

\newcommand{\eq}{\begin{equation}}
\newcommand{\eeq}{\end{equation}}

\def\XXint#1#2#3{{\setbox0=\hbox{$#1{#2#3}{\int}$ }
\vcenter{\hbox{$#2#3$ }}\kern-.6\wd0}}

\newcommand{\Union}{\mathop{\bigcup}\limits}

 

\errorcontextlines=0 \numberwithin{equation}{section}
\theoremstyle{plain}
\newtheorem{theorem}{Theorem}[section]
\newtheorem{lemma}[theorem]{Lemma}

\newtheorem{proposition}[theorem]{Proposition}
\newtheorem{definition}[theorem]{Definition}
\newtheorem{remark}[theorem]{Remark}
\newtheorem{corollary}[theorem]{Corollary}

\title[Spectral analysis of a complex Schr\"odinger operator]{Spectral
  analysis of a complex Schr\"odinger operator in the semiclassical
  limit} 

\author{Yaniv Almog}
\address{Department of Mathematics\\
  Louisiana State University\\
  Baton Rouge, Louisiana 70803}
\email[Y.~Almog]{almog@math.lsu.edu}
\author{Rapha\"{e}l Henry}
\address{Laboratoire de Math\'ematiques d'Orsay\\
Univ. Paris-Sud, CNRS, Universit\'e Paris-Saclay\\
91405 Orsay, France.}
\email[R.~Henry]{raphael.henry@math.u-psud.fr}
\date{}

\begin{document}
\bibliographystyle{siam}

\maketitle

\begin{abstract}
  We consider the Dirichlet realization of the operator $-h^2\Delta+iV$ in
  the semi-classical limit $h\to0$, where $V$ is a smooth real
  potential with no critical points. For a one dimensional setting, we
  obtain the complete asymptotic expansion, in powers of $h$, of each
  eigenvalue. In two dimensions we obtain the left margin of the
  spectrum, under some additional assumptions.
\end{abstract}
\section{Introduction}
\label{sec:1}
We consider the operator
\begin{subequations}
  \label{eq:1}
  \begin{equation}\label{eq:1a}
\A_h = -h^2\Delta + iV \,,
\end{equation}
defined on
\begin{equation}
  D(\A_h)=H_0^1(\Omega,\C)\cap H^2(\Omega,\C)\,,
\end{equation}
\end{subequations}
where $\Omega$ is a bounded domain in $\mathbb{R}^2$.\\
We seek an approximation for $\inf \Re\sigma(\A_h)$ in the limit $h\to0$.
The domain $\Omega$ is smooth, i.e., $\partial\Omega\subset C^3$ and the potential
$V$ is at least in $C^3(\bar\Omega,\R)$. Let $\partial\Omega_\perp$ denote a subset of
$\partial\Omega$ where $\nabla V\perp\partial\Omega$. Note that in view of the continuity of $V$
on $\partial\Omega$, we must have $\partial\Omega_\perp\neq\emptyset$. Let $x_0\in\partial\Omega_\perp$ satisfy
\begin{displaymath}
  c_m=|\nabla V(x_0)| = \min_{x\in\partial\Omega_\perp}|\nabla V(x)| \,. 
\end{displaymath}

Denote by $\Sg$ the set 
\begin{equation}\label{eq:2}
  \Sg=\{x\in\partial\Omega_\perp\,: \,|\nabla V(x)| = |\nabla V(x_0)|  \,, \;V(x)=V(x_0)\}\,.
\end{equation}
(Note that in the case where $x_0$ is not unique, $\Sg$ depends on the
choice of $x_0$.) 
For every $x\in\Sg$ set
\begin{displaymath}
  c(x) = \nabla V(x)\cdot\nu(x) = \pm c_m \,,
\end{displaymath}
and 
\begin{displaymath}
  \alpha(x) = t\cdot D^2V(x)t- \kappa(x)\frac{\partial V}{\partial\nu}(x)  \quad       t\cdot\nu(x)=0\,, \;
      |t|=1 \,,
\end{displaymath}
where $\nu$ is the outward normal and $\kappa$ denotes the local curvature.
(Note that $\alpha=\partial^2V/\partial s^2$ where $s$ is the arclength on $\partial\Omega$.) We
now assume that
\begin{equation}
  \label{eq:3}
\alpha(x)c(x) > 0 \quad \forall x\in\Sg \,.
\end{equation}
Without any loss of generality we may then assume $\alpha(x)>0$ in $\Sg$,
otherwise we may consider $\bar{\A}_h$ instead of $\A_h$. 

The spectral analysis of \eqref{eq:1} has several applications in
Mathematical Physics, among them are the Orr-Sommerfeld equations in
fluid dynamics \cite{sh03}, the Ginzburg-Landau equation in the
presence of electric current (when magnetic field effects are
neglected), and the null controllability of Kolomogorov type equations
\cite{bhhr14}.  In \cite{al08,hen15} it has been established that
\begin{equation}
  \label{eq:4}
\liminf_{h\to0}h^{-2/3}\inf \Re\sigma(\A_h) \geq \frac{|\mu_1|}{2}c_m^{2/3} \,,
\end{equation}
where $\mu_1$ is the rightmost zero of Airy's function
\cite{abst72}. 

We note that \eqref{eq:4} has been obtained without the need to
assume \eqref{eq:3}. In the present contribution we seek an upper
bound for $\inf \Re\sigma(\A_h)$. It is to this end that we make
that additional assumption. Our main result is the following
\begin{theorem}
  \label{thm:1}
  Let $\A_h$ denote the Dirichlet realization of a Schr\"odinger
  operator with a purely imaginary potential $V\in C^3(\Omega,\R)$,
  satisfying $\nabla V\neq0$ in $\bar{\Omega}$, given by
  \eqref{eq:1}. Suppose that $V$ satisfies \eqref{eq:3}.
  Then, there
  exists $\lambda(h) \in \sigma(\A_h)$ satisfying
  \begin{equation}
    \label{eq:5}
\Big|\lambda-  iV(x_0)- e^{i\pi/3}|\mu_1|(c_mh)^{2/3}-\sqrt{2\alpha}e^{i\frac{\pi}{4}}
h\Big|\sim o(h) \quad \text{as }h\to0\,,
  \end{equation}
  where $\alpha = \alpha(x_0)\,$.
\end{theorem}
An immediate corollary follows
\begin{corollary}
  Under the above assumptions we have that
  \begin{equation}
\label{eq:53}
    \lim_{h\to0}h^{-2/3}\inf \Re\sigma(\A_h) = \frac{|\mu_1|}{2}c_m^{2/3} \,,
  \end{equation}
\end{corollary}

\begin{remark}
  While we do not prove that here, it appears that \eqref{eq:53} can
  be extended to higher dimensions. Let $D^2_\parallel V$ denote the
  Hessian matrix of $V$ with respect to a local curvilinear coordinate system
  defined on $\partial\Omega$ (including, of course, curvature
  effects). Suppose that $D^2_\parallel V(x)$ is either positive or
  negative. Then, we set $\alpha$ in the following manner
  \begin{displaymath}
    \alpha(x) = \sign \big(D^2_\parallel V(x)\big)\inf_{
      \begin{subarray}\strut
        t\cdot\nu(x)=0 \\
        |t|=1
      \end{subarray}}
|t\cdot D^2_\parallel V(x)t| \,,
  \end{displaymath}
and assume  \eqref{eq:3} once again.
\end{remark}

\begin{remark}
  Let $\A^N_h$ denote the Neumann realization of $\A_h$. By using the
  same techniques as in the sequel, one can obtain an upper bound for
  $\inf \Re\sigma(\A^N_h)$. In this case, $\mu_1$ will be replaced by the
  rightmost critical point Airy's function.
\end{remark}

Finally, we note that it has been established in \cite{hen15} that for
all $\epsilon>0$ there exist positive $M_\epsilon$ and $h_\epsilon$ such that for all
$h\in(0,h_\epsilon)$ we have the following upper bound for the semigroup
assciated with $-\A_h$,
\begin{displaymath}
  \|e^{-t\A_h}\|\leq M_\epsilon \exp \{ -(c_m^{2/3}|\mu_1|/2-\epsilon)h^{2/3}t\} \,.
\end{displaymath}
From \eqref{eq:5} we can now establish that for some positive $M$, $C$
and $h_0$ the following lower bound for
the semigroup holds for all $h\in(0,h_0)$
\begin{displaymath}
   \|e^{-t\A_h}\|\geq  M \exp \Big\{ -c_m^{2/3}\frac{|\mu_1|}{2}h^{2/3}(1+Ch^{1/3})t\Big\} \,.
\end{displaymath}

The rest of this contribution is arranged as follows: in the next
section we consider a one-dimensional version of \eqref{eq:1}.
Assuming that $V\in C^\infty([0,a],\R)$ we obtain the complete
asymptotic expansion, as $h\to0$, of any eigenvalue $\lambda_k\in\sigma(\A_h)$
($k$ is fixed in the limit). In Section \ref{sec:2} we construct the quasimode
associated with the eigenvalue given in \eqref{eq:5}, and in the
last section provide a rigorous proof of Theorem \ref{thm:1}.

\section{The one-dimensional case}
\label{s:1D}

\subsection{Statement of the results}

Let $a>0$ and $V\in\mathcal{C}^\infty\big([0,a] ; \mathbb{R}\big)\,$ such that $V$ has no critical point
in $[0,a]\,$.  Consider then the one-dimensional Schr\"odinger operator
$\mathcal{A}_h$ defined on $(0,a)$ by
\begin{displaymath}
 \mathcal{A}_h = -h^2\frac{d^2}{dx^2}+i\big(V-V(0)\big)\,,
\end{displaymath}
with domain
\begin{displaymath}
 D(\mathcal{A}_h) = H_0^1([0,a],\C)\cap H^2([0,a],\C)\,.
\end{displaymath}
The main result we prove in this section is the following:
\begin{theorem}\label{thm:1D}
Assume that, for all $x\in[0,a]\,$, $V'(x)\neq0\,$. Then,
for all $n\geq1\,$, there exists a complex sequence $(\alpha_{j,n})_{j\geq1}$ and an
eigenvalue $\lambda_n(h)\in\sigma(\mathcal{A}_h)$ such that, as $h\to0\,$,
\begin{equation}\label{devBord}
h^{-2/3}\lambda_n(h) \underset{h\to0}{\sim} e^{\sigma i\pi/3}|V'(0)|^{2/3}|\mu_n|
+\sum_{j=1}^{+\infty}\alpha_{j,n}h^{2j/3} + \mathcal{O}(h^\infty)\,,
\end{equation}
where $\sigma$ is the (constant) sign of the function $V'\,$.\\
\end{theorem}
Similarly, one could also prove the existence of another sequence
$(\nu_n(h))_{n\geq1}$ of eigenvalues satisfying an asymptotic expansion
of the form
\begin{equation}\label{devBord2}
\nu_n(h) \underset{h\to0}{\sim}  i\big(V(a)-V(0)-a\big) + e^{\sigma i\pi/3}|V'(a)|^{2/3}|\mu_n|h^{2/3}
+\sum_{j=1}^{+\infty}\beta_{j,n}h^{2(j+1)/3} + \mathcal{O}(h^\infty)
\end{equation}
by applying the transformation $x\to a-x$. Similar results have
previously been obtained in the particular cases $V(x) = x$ and 
$V(x) = x^2$, see \cite{sh03}
and \cite{bhhr14}. 

\begin{remark}
  Theorem \ref{thm:1D} esablishes existence of two sequences of
  eigenvalues of $\A_h$, respectively obeying \eqref{devBord} and
  \eqref{devBord2}. The fact that these sequences constitute the
  entire spectrum of $\A_h$ for $\Re \lambda \leq Mh^{2/3}$ for any positive
  $M$ follows from \cite[Proposition 6.1]{hen15}.
\end{remark}

Let $\varepsilon=h^{2/3}$. It is more convenient to obtain the spectrum of $\A_h$ by first
applying the dilation operator $U:L^2(0,a)\to L^2(0,a/\varepsilon)$ defined by
\begin{displaymath}
  (Uu)(\cdot/\varepsilon)=u(\cdot) \,.
\end{displaymath}
Let
\begin{displaymath}
 V_\varepsilon(x) = \frac{V(\varepsilon x)}{\varepsilon}\,.
\end{displaymath}
Then by applying the above dilation we obtain
\begin{equation}
\label{eq:6}
 \frac{1}{\varepsilon}U^{-1}\A_hU= \mathcal{A}_\varepsilon = -\frac{d^2}{dx^2}+i\left(V_\varepsilon-\frac{V(0)}{\varepsilon}\right)\,,
\end{equation}
defined on
\begin{displaymath}
 D(\mathcal{A}_\varepsilon) = (H_0^1\cap H^2)\big((0,a/\varepsilon),\C\big)\,.
\end{displaymath}

\subsection{Quasimode construction}\label{ss:1D_quasimode}
In the following we construct quasimodes and approximate eigenvalues
for $\mathcal{A}_\varepsilon$ in the neighborhood of the boundary point
$x=0\,$.  In particular, we obtain the asymptotic expansion (\ref{devBord}) for
each approximate eigenvalue. 

\begin{proposition}
 Assume that, for all $x\in[0,a]\,$, $V'(x)\neq0\,$. Let $n\geq1$ and $\sigma$ denote the sign of $V'\,$.
 Then there exists $\psi_\varepsilon\in\mathcal{D}(\mathcal{A}_\varepsilon)$
and a complex sequence $(\nu_j)_{j\geq2}$
 such that
 \begin{equation}
\label{eq:7}
 \big\|(\mathcal{A}_\varepsilon-\nu(\varepsilon))\psi_\varepsilon\big\| = \mathcal{O}(\varepsilon^\infty)\|\psi_\varepsilon\|\,,
 \end{equation}
 where
 \begin{equation}
 \nu(\varepsilon) = e^{\sigma i\pi/3}|V'(0)|^{2/3}|\mu_n| + \sum_{j=1}^{+\infty}\nu_j\varepsilon^j
 + \mathcal{O}(\varepsilon^\infty)
 \end{equation}
 as $\varepsilon\to0\,$.
\end{proposition}
\begin{proof}

We approximate $\mathcal{A}_\varepsilon$ at any order $N$ by the operator
\begin{displaymath}
A_N(\varepsilon) = A_0 + \sum_{j=1}^NV_j\varepsilon^j~~~\textrm{ on }~~~(0,+\infty)\,,
\end{displaymath}
where
\begin{eqnarray*}
  A_0 = -\frac{d^2}{dx^2} + i\beta_0x\,,&\beta_0 = V'(0)\,,&\\
  V_j = i\beta_jx^{j+1}\,,&\beta_j = \frac{V^{(j+1)}(0)}{(j+1)!}\,,& j\in\mathbb{N}\,.
\end{eqnarray*}
Then, for all $N\geq1\,$,
we look for a quasimode $u^N(x,\varepsilon)$ and an approximate eigenvalue $\lambda^N(\varepsilon)$ in the form
\begin{equation}
\label{eq:8}
 u^N(x,\varepsilon) = \sum_{j=0}^N u_j(x) \varepsilon^j\,,~~~\lambda^N(\varepsilon) = \sum_{j=0}^N \lambda_j\varepsilon^j\,,
\end{equation}
satisfying
\begin{displaymath}
 \Big(A_0+\sum_{j=1}^N V_j\varepsilon^j\Big)u^N(x,\varepsilon) = \lambda^N(\varepsilon)u^N(x,\varepsilon) + \OO(\varepsilon^{N+1})\,.
\end{displaymath}
To this end, we need to successively solve the following equations:
\begin{eqnarray}
 (A_0-\lambda_0)u_0 & = & 0\,, \nonumber \\
 (A_0-\lambda_0)u_1 & = & -(V_1-\lambda_1)u_0\,,\nonumber  \\
 \vdots & & \\
 (A_0-\lambda_0)u_k & = & -\sum_{j=1}^k(V_j-\lambda_j)u_{k-j}\,,~~~k=1,\dots,N\,. \label{eqK}
\end{eqnarray}
Consider the first equation. If $\beta_0>0\,$, we can use the scale change
$x\mapsto \beta_0^{1/3}x$ and the well-known properties of the complex Airy
operator \cite{al08} to obtain
\begin{displaymath}
 \sigma(A_0) = \big\{\beta_0^{1/3}\mu_ne^{-2i\pi/3} : n\in\mathbb{N}\big\}\,,
\end{displaymath}
where $\mu_n$ denotes the $n$-th zero of the Airy function $Ai\,$. The associated eigenfunctions are
\begin{displaymath}
 x\mapsto Ai(\beta_0^{1/3}e^{i\pi/6}x+\mu_n)\,.
 \end{displaymath}
 If $\beta_0<0\,$, then the operator $A_0$ is the adjoint
 of $-\frac{d^2}{dx^2}+i|\beta_0|x\,$. Hence,
\begin{displaymath}
 \sigma(A_0) = \big\{|\beta_0|^{1/3}\mu_ne^{+2i\pi/3} : n\in\mathbb{N}\big\}\,,
\end{displaymath}
and the eigenfunctions are given by
\begin{displaymath}
 x\mapsto \overline{Ai(\beta_0^{1/3}e^{i\pi/6}x+\mu_n)}\,.
\end{displaymath}
Therefore, for any fixed $n\in\mathbb{N}\,$, we choose
\begin{equation}
 \lambda_0 = \lambda_{0,n} = |\beta_0|^{1/3}\mu_ne^{\sigma 2i\pi/3}\,,
\end{equation}
and $u_0 = u_{0,n}$ to be a corresponding eigenfunction.\\

Next, consider the second equation. To ensure the existence of a
$u_1\,$, we first select $\lambda_1$ such that
\begin{displaymath}
 (V_1-\lambda_1)u_0 \in \Im(A_0-\lambda_0) = \ker(A_0^*-\bar\lambda_0)^\perp\,.
\end{displaymath}
Since $\ker(A_0^*-\bar\lambda_0) = \langle\bar u_0\rangle\,$ we may conclude that
\begin{equation}\label{condLambda1}
\lambda_1\int_{\mathbb{R}_+}u_0(x)^2dx = i\beta_1\int_{\mathbb{R}_+}x^2u_0(x)^2dx\,.
\end{equation}
Furthermore, as $u_0(x) = Ai(\beta_0^{1/3}e^{i\pi/6}x+\mu_n)$ (respectively
$u_0(x) = \overline{Ai(\beta_0^{1/3}e^{i\pi/6}x+\mu_n)}$) for $\beta_0>0$
(respectively  $\beta_0<0$), Cauchy Theorem and the decay of $Ai$ in the sector
$\{|\arg z|\leq \pi/3\}$ immediately yields
  \begin{displaymath}
    \int_{\mathbb{R}_+}u_0(x)^2dx =
    e^{-i\pi/6}\int_{\mathbb{R}_+}Ai^2(\beta_0^{1/3}x+\mu_n) dx\neq0\,.
  \end{displaymath}
Thus, we may select
\begin{equation}
\lambda_1 =
i\beta_1\frac{\int_{\mathbb{R}_+}x^2u_0(x)^2dx}{\int_{\mathbb{R}_+}u_0(x)^2dx}=
  i\beta_1e^{-i\pi/3}\frac{\int_{\mathbb{R}_+}x^2Ai^2(\beta_0^{1/3}x+\mu_n)dx}{\int_{\mathbb{R}_+}Ai^2(\beta_0^{1/3}x+\mu_n)dx}
  \,,
\end{equation}
and there exists $u_1\in D(A_0)$ such that
\begin{displaymath}
 (A_0-\lambda_0)u_1 = -V_1u_0\,.
\end{displaymath}
Assuming that the first $k$  equations are solved by
$\lambda_0,\dots,\lambda_{k-1}\,$, $u_0,\dots,u_{k-1}\,$, we have to choose such
$\lambda_k$ so that a solution $u_k$ to the $(k+1)$-th
equation exists. It easily follows that the solvability condition is
\begin{displaymath}
 -\sum_{j=1}^k(V_j-\lambda_j)u_{k-j}\in \ker(A_0^*-\bar\lambda_0)^\perp\,,
\end{displaymath}
yielding
\begin{equation}\label{exprLambdak}
\lambda_k = \frac{1}{\int_{\mathbb{R}_+}u_0(x)^2dx}\left(\sum_{j=1}^{k-1}\int_{\mathbb{R}_+}\big(i\beta_jx^{j+1}-\lambda_j)u_{k-j}(x)u_0(x)dx
+i\beta_k\int_{\mathbb{R}_+}x^{k+1}u_0(x)^2dx\right)\,.
\end{equation}
For this value of $\lambda_k\,$, there exists $u_k\in\mathcal{D}(A_0)$
satisfying (\ref{eqK}). Invoking inductive arguments, we assume that
each function $u_0,\dots,u_{k-1}$ is in $\mathcal{S}(\mathbb{R}_+)\,$.
Then, it easily follows that $u_k\in\mathcal{S}(\mathbb{R}_+)\,$. We
can then set $u(x,\varepsilon)$ and $\lambda(\varepsilon)$ to be some appropriate Borel sums
of the formal series $\sum u_j(x)\varepsilon^j$ and $\sum\lambda_j \varepsilon^j\,$,
respectively.\\

We now construct from $u(\cdot,\varepsilon)$ a quasimode satisfying the desired
boundary conditions.  Let $c_0>0$ and
$\chi\in\mathcal{C}_0^\infty\big((-c_0,c_0) ; [0,1]\big)$ be such that $\chi(y)
= 1$ for all $y\in[-c_0/2,c_0/2]$, and such that $\chi', \chi''$ are
bounded. We set
\[
 \chi_\varepsilon(x) = \chi(\varepsilon^{1-\rho}x)\,.
\]
Then, for $p=1,2\,$, we have
\begin{equation}\label{suppChi1}
\mathbb{R}_+\cap\textrm{Supp }\chi_\varepsilon^{(p)} \subset [c_0\varepsilon^{\rho-1}/2,c_0\varepsilon^{\rho-1}]\,,
\end{equation}
and
\begin{equation}\label{suppChi2}
\sup_{x\in\mathbb{R}}\big|\chi_\varepsilon^{(p)}(x)\big| = \mathcal{O}\big(\varepsilon^{p(1-\rho)}\big)\,.
\end{equation}
We next define
\begin{displaymath}
\psi_\varepsilon(x) =
  \mathbf{1}_{\mathbb{R}_+}(x)\chi_\varepsilon(x)u(x,\varepsilon)\,.
\end{displaymath}
Then, we write
\begin{displaymath}
 \mathcal{A}_\varepsilon = A_0 + \sum_{j=1}^NV_j(x)\varepsilon^j + \frac{1}{\varepsilon}R_{N+1}(\varepsilon,x)\,,
\end{displaymath}
where $R_{N+1}$ denotes the remainder term in the $(N+1)$-th order
Taylor expansion of $V$ near $x=0$ (so that $\varepsilon^{-1}R_{N+1}(\varepsilon x)$ is of
order $\mathcal{O}(\varepsilon^{N+1})\,$).\\
Then, we have
\begin{equation}\label{decomp1}
 \big(\mathcal{A}_\varepsilon-\lambda(\varepsilon)\big)\psi_\varepsilon  =  
 \chi_\varepsilon\big(\mathcal{A}_\varepsilon-\lambda(\varepsilon)\big)u(\cdot,\varepsilon) + [\mathcal{A}_\varepsilon,\chi_\varepsilon]u(\cdot,\varepsilon)\,.
\end{equation}
We seek an estimate for both terms on the right-hand side.
Consider the first term, for which we have
\begin{equation}
\label{eq:54}
  \left\|\chi_\varepsilon\big(\mathcal{A}_\varepsilon-\lambda(\varepsilon)\big)u(\cdot,\varepsilon)\right\|  \leq 
  \left\|\left(A_0+\sum_{j=1}^NV_j\varepsilon^j-\lambda(\varepsilon)\right)u(\cdot,\varepsilon)\right\| + \left\|\varepsilon^{-1}R_{N+1}(\varepsilon,\cdot)u(\cdot,\varepsilon)\right\| \,.
\end{equation}
By the construction of $u$ and $\lambda$, the first term on the
right-hand side is of order $\mathcal{O}(\varepsilon^{N+1})$. Furthermore,
there exists $c_N>0$ such that
\begin{equation}
 \left\|\varepsilon^{-1}R_{N+1}(\varepsilon\cdot)u(\cdot,\varepsilon)\right\| \leq c_N\varepsilon^{N+1}\|x^{N+2}u(\cdot,\varepsilon)\|
  = \mathcal{O}(\varepsilon^{N+1})\,,
\end{equation}
where we made use of the fact that $u(\cdot,\varepsilon)\in\mathcal{S}(\mathbb{R})\,$. Therefore,
there exists $K_N>0$ such that
\begin{equation}\label{estRHS1}
\left\|\chi_\varepsilon\big(\mathcal{A}_\varepsilon-\lambda(\varepsilon)\big)u(\cdot,\varepsilon)\right\| \leq K_N\varepsilon^{N+1}\,.
\end{equation}
Consider, next, the commutator term in (\ref{decomp1}). Since
$u(\cdot,\varepsilon)\in\mathcal{S}(\mathbb{R})\,$, (\ref{suppChi1}) and
(\ref{suppChi2}) yield
\begin{equation}\label{estRHScommut}
\left\|[\mathcal{A}_\varepsilon,\chi_\varepsilon]u(\cdot,\varepsilon)\right\|  \leq  2\|\chi_\varepsilon'\pa_xu(\cdot,\varepsilon)\| + \|\chi_\varepsilon''u(\cdot,x)\|
= \mathcal{O}(\varepsilon^\infty)\|\psi_\varepsilon\|\,.
\end{equation}
Finally, by (\ref{decomp1}), (\ref{estRHS1}) and (\ref{estRHScommut}), we have
\begin{displaymath}
 \left\|\big(\mathcal{A}_\varepsilon-\lambda(\varepsilon)\big)\psi_\varepsilon\right\| =
 \mathcal{O}(\varepsilon^\infty)\|\psi_\varepsilon\|\,.
\end{displaymath}
\end{proof}

\subsection{Proof of Theorem \ref{thm:1D}}\label{ss:1D_proof}
Once the quasimodes associated with the approximate eigenvalues
(\ref{devBord}) have been found, it remains necessary to prove that
such eigenvalues indeed exist in $\sigma(\A_h)$.

\begin{lemma}
  Let $n\in\N$ and $\lambda_n$ be given by the expansion \eqref{devBord}.  Let
  $\lambda=\lambda_n+re^{i\theta}$ where $\theta\in[0,2\pi)$. Then for $\alpha\in(1,4/3)$, there exist $\delta>0$, $\varepsilon_0>0$
  and $C>0$  such that for any $\varepsilon\in(0,\varepsilon_0)$ and $r$ satisfying $\varepsilon^{\alpha}<r<\delta$, we have
  \begin{equation}
    \label{eq:9}
 \|(\mathcal{A}_\varepsilon-\lambda)^{-1}\|\leq \frac{C}{r} \,.
 \end{equation}
  \end{lemma}
  \begin{proof}
    Let $f\in L^2(0,a/\varepsilon)$ and $u\in D(\A_\varepsilon)$ satisfy
    \begin{equation}
\label{eq:10}
      (\mathcal{A}_\varepsilon-\lambda)u=f \,.
    \end{equation}
Let $\tilde{\chi}_\varepsilon$ satisfy
\begin{displaymath}
  \chi_\varepsilon^2+\tilde{\chi}_\varepsilon^2 =1
\end{displaymath}
and
\begin{equation}
\sup_{x\in\mathbb{R}}\big|\nabla\tilde{\chi}_\varepsilon(x)\big| = \mathcal{O}\big(\varepsilon^{(1-\rho)}\big)\,.
\end{equation}
Taking the inner product in $L^2(0,a/\varepsilon)$ of \eqref{eq:10} with $\tilde{\chi}_\varepsilon^2 u$ we obtain from
the real part
\begin{displaymath}
  \|\nabla(\tilde{\chi}_\varepsilon u)\|_2^2 = \Re\langle\tilde{\chi}_\varepsilon u,\tilde{\chi}_\varepsilon f\rangle+
  \|u\nabla\tilde{\chi}_\varepsilon\|_2^2+\Re\lambda\|\tilde{\chi}_\varepsilon u\|_2^2 \,.
\end{displaymath}
Hence,
\begin{equation}
  \label{eq:11}
\|\nabla(\tilde{\chi}_\varepsilon u)\|_2\leq  C\big(\varepsilon^{-(1-\rho)}\|\tilde{\chi}_\varepsilon f\|_2+\|\tilde{\chi}_\varepsilon u\|_2
+\varepsilon^{1-\rho}\|u\|_2\big)\,.
\end{equation}

From the imaginary part of the above inner product we obtain that
\begin{displaymath}
  \langle\tilde{\chi}_\varepsilon(V_\varepsilon-\varepsilon^{-1}V(0))u,\tilde{\chi}_\varepsilon u\rangle
  =\Im\langle\tilde{\chi}_\varepsilon u,\tilde{\chi}_\varepsilon f\rangle+\Im\langle\nabla(\tilde{\chi}_\varepsilon u),u\nabla\tilde{\chi}_\varepsilon\rangle+\Im\lambda\|\tilde{\chi}_\varepsilon u\|_2^2 \,.
\end{displaymath}
Since
\begin{displaymath}
 \min_{x\in(0,a/\varepsilon)} |\tilde{\chi}_\varepsilon(V_\varepsilon-\varepsilon^{-1}V(0))|\geq C\varepsilon^{\rho-1} \,,
\end{displaymath}
We obtain that
\begin{displaymath}
  \|\tilde{\chi}_\varepsilon u\|_2^2 \leq C\varepsilon^{1-\rho}\big[\|\tilde{\chi}_\varepsilon u\|_2^2
  +\|\tilde{\chi}_\varepsilon f\|_2^2+  \varepsilon^{2(1-\rho)}\|\nabla(\tilde{\chi}_\varepsilon u)\|_2^2 +\|u\|_2^2\big] \,.
\end{displaymath}
With the aid of \eqref{eq:11} we then obtain
\begin{equation}
\label{eq:12}
  \|\tilde{\chi}_\varepsilon u\|_2\leq C\varepsilon^{(1-\rho)/2}(\|u\|_2 +\|f\|_2) \,.
\end{equation}

We next seek an estimate for $\|\chi_\varepsilon u\|_2$. 
To this end we write
\begin{equation}
\label{eq:13}
  (A_0-\lambda)(\chi_\varepsilon u)=\chi_\varepsilon f -
  i\Big(V_\varepsilon-\frac{V(0)}{\varepsilon}-\beta_0x\Big) \chi_\varepsilon u +
  [\mathcal{A}_\varepsilon,\chi_\varepsilon]u \,.
\end{equation}
Denote by $v_n$ the eigenfunction of $A_0$ associated with the
eigenvalue $e^{i\pi/3}\beta_0^{1/3}\mu_n$. 
For any $g\in L^2(0,a/\varepsilon)$ let
\begin{displaymath}
  \Pi_ng=\langle\bar{v}_n,g\rangle v_n \,.
\end{displaymath}
Let further
\begin{displaymath}
  w_n = (I-\Pi_n)(\chi_\varepsilon u) \,.
\end{displaymath}
By \eqref{eq:13} we easily obtain that
\begin{displaymath}
  (A_0-\lambda)w_n = (I-\Pi_n)\Big(\chi_\varepsilon f -
  i\Big(V_\varepsilon-\frac{V(0)}{\varepsilon}-\beta_0x\Big) \chi_\varepsilon u +
  [\mathcal{A}_\varepsilon,\chi_\varepsilon]u \Big)\,.
\end{displaymath}
By the Riesz-Schauder theory for compact operators (cf. \cite{ag65}
for instance) we have that
\begin{displaymath}
  (A_0-\lambda)^{-1} = \frac{\Pi_n}{\lambda-\lambda_{0,n}} + T_n(\lambda) \,,
\end{displaymath}
where $T_n(\lambda)$ is holomorphic, and hence bounded, in some
fixed neighborhood of $\lambda_{0,n}$. Consequently, there exists
$C(n,\beta_0)$ such that $\|(A_0-\lambda)^{-1}(I-\Pi_n)\|\leq C$, and  therefore, 
\begin{multline*}
  \| w_n\|_2\leq C \Big\|\Big(\chi_\varepsilon f -
  i\Big(V_\varepsilon-\frac{V(0)}{\varepsilon}-\beta_0x\Big) \chi_\varepsilon u +
  [\mathcal{A}_\varepsilon,\chi_\varepsilon]u \Big)\Big\|_2\\\leq C \Big(\|f\|_2 +
  \Big\|\Big(V_\varepsilon-\frac{V(0)}{\varepsilon}-\beta_0x\Big) \chi_\varepsilon u\Big\|_2 +
 \| [\mathcal{A}_\varepsilon,\chi_\varepsilon]u \|_2\Big)\,.
\end{multline*}
Hence,
\begin{displaymath}
  \| w_n\|_2\leq C \big(\|f\|_2
  +[\varepsilon^{2\rho-1}+\varepsilon^{2(1-\rho)}]\|u\|_2+\varepsilon^{1-\rho}\|\nabla u\|_2\big)\,,
\end{displaymath}
and since
\begin{equation}
\label{eq:14}
  \|\nabla u\|_2^2= \Re\langle u,f\rangle+ \Re\lambda\|u\|_2^2\,,
\end{equation}
we obtain that
\begin{equation}
  \label{eq:15}
\| w_n\|_2\leq C \big(\|f\|_2
  +[\varepsilon^{2\rho-1}+\varepsilon^{1-\rho}]\|u\|_2)\,.
\end{equation}

To complete the proof, we seek an estimate for $\Pi_n(\chi_\varepsilon u)$. Taking
the inner product of \eqref{eq:13} with $\chi_\varepsilon\bar{v}_n$ yields
\begin{multline}
\label{eq:16}
  (e^{i\pi/3}\beta_0^{1/3}\mu_n-\lambda)\gamma_n=\langle\bar{v}_n,f\rangle+\langle[A_0,\chi_\varepsilon]\bar{v}_n,\chi_\varepsilon u\rangle- \langle\tilde{\chi}_\varepsilon\bar{v}_n,\tilde{\chi}_\varepsilon f\rangle+\\
  i\Big\langle{\bar v}_n,\Big(V_\varepsilon-\frac{V(0)}{\varepsilon}-\beta_0x\Big) \chi_\varepsilon u \Big\rangle
  +\langle\chi_\varepsilon\bar{v}_n,[A_0,\chi_\varepsilon]u \rangle +\\
  (e^{i\pi/3}\beta_0^{1/3}\mu_n-\lambda)\langle\tilde{\chi}_\varepsilon v_n,\tilde{\chi}_\varepsilon u\rangle -
  i\Big\langle(1-\chi_\varepsilon){\bar v}_n,\Big(V_\varepsilon-\frac{V(0)}{\varepsilon}-\beta_0x\Big)\chi_\varepsilon u \Big\rangle\,,
  \end{multline}
  where
\begin{displaymath}
  \gamma_n=\langle\bar{v}_n,\chi_\varepsilon u\rangle \,.
\end{displaymath}
By the exponential decay of $v_n$ and \eqref{eq:14} we have that
\begin{multline}
\label{eq:17}
  \Big| \langle[A_0,\chi_\varepsilon]\bar{v}_n,\chi_\varepsilon u\rangle- \langle\tilde{\chi}_\varepsilon\bar{v}_n,\tilde{\chi}_\varepsilon f\rangle +
  (e^{i\pi/3}\beta_0^{1/3}\mu_n-\lambda)\langle\tilde{\chi}_\varepsilon v_n,\tilde{\chi}_\varepsilon u\rangle -\\
  i\Big\langle(1-\chi_\varepsilon){\bar v}_n,\Big(V_\varepsilon-\frac{V(0)}{\varepsilon}-\beta_0x\Big)
  \chi_\varepsilon u \Big\rangle\Big| \leq Ce^{-\varepsilon^{-3(1-\rho)/2}}(\|u\|_2+\|f\|_2)\,.
\end{multline}
We next write
\begin{multline*}
   \Big\langle{\bar v}_n,\Big(V_\varepsilon-\frac{V(0)}{\varepsilon}-\beta_0x\Big) \chi_\varepsilon u \Big\rangle=
   \varepsilon\gamma_n\langle{\bar v}_n,\beta_1x^2v_n\rangle \\ +  \Big\langle{\bar v}_n,\Big(V_\varepsilon-\frac{V(0)}{\varepsilon}-\beta_0x\Big)
   w_n \Big\rangle+
  \gamma_n \Big\langle{\bar v}_n,\Big(V_\varepsilon-\frac{V(0)}{\varepsilon}-\beta_0x-\varepsilon\beta_1x^2\Big)v_n\rangle \,.
\end{multline*}
We now observe that
\begin{displaymath}
  \Big\|{\bar v}_n\Big(V_\varepsilon-\frac{V(0)}{\varepsilon}-\beta_0x\Big)\Big\|_2 \leq C\varepsilon \,,
\end{displaymath}
and that
\begin{displaymath}
   \Big|\Big\langle{\bar v}_n,\Big(V_\varepsilon-\frac{V(0)}{\varepsilon}-\beta_0x-\varepsilon\beta_1x^2\Big)v_n\rangle\Big|\leq C\varepsilon^2 \,.
\end{displaymath}
As $|\gamma_n|\leq \|u\|_2$, we obtain with the aid of \eqref{eq:15} that
\begin{displaymath}
   \Big|\Big\langle{\bar v}_n,\Big(V_\varepsilon-\frac{V(0)}{\varepsilon}-\beta_0x\Big) \chi_\varepsilon u
   \Big\rangle- \varepsilon\gamma_n\langle{\bar v}_n,\beta_1x^2v_n\rangle \Big|\leq C\varepsilon(\|f\|_2 +[\varepsilon^{2\rho-1}+\varepsilon^{1-\rho}]\|u\|_2)\,.
\end{displaymath}
Substituting the above, together with \eqref{eq:17} into \eqref{eq:16}
yields
\begin{displaymath}
  |(e^{i\pi/3}\beta_0^{1/3}\mu_n+i\varepsilon\gamma_n\langle{\bar v}_n,\beta_1x^2v_n\rangle-\lambda)\gamma_n| \leq C(\|f\|_2 +
  [\varepsilon^{2\rho}+\varepsilon^{2-\rho}]\|u\|_2)
\end{displaymath}
Consequently, we must have
\begin{equation}
\label{eq:18}
  |\gamma_n|\leq \frac{C}{r}(\|f\|_2 + [\varepsilon^{2\rho}+\varepsilon^{2-\rho}]\|u\|_2)\,.
\end{equation}

We now choose $\rho=2/3$. Since
\begin{displaymath}
  \|u\|_2\leq C(|\gamma_n| + \|w_n\|_2+ \|\tilde{\chi}_\varepsilon u\|_2)\,,
\end{displaymath}
\eqref{eq:9} easily follows from \eqref{eq:12}, \eqref{eq:15}, and \eqref{eq:18}.
  \end{proof}

\begin{lemma}
  \label{lem:existunique}
Let $1<\alpha<4/3$. Let further
\begin{equation}
\label{eq:19}
  \Lambda_{n,N}(\varepsilon)=e^{\sigma i\pi/3}|\beta_0|^{2/3}|\mu_n|
+\sum_{j=1}^{N}\alpha_{j,n}\varepsilon^j \,.
\end{equation}
Then, for sufficiently small $\varepsilon$ there exists $\lambda_n(\varepsilon)$ such that
\begin{equation}
  \label{eq:20}
\sigma(\A_\varepsilon)\cap B(\Lambda_{n,1},2\varepsilon^\alpha)=\{\lambda_n(\varepsilon)\}\,.
\end{equation}
Furthermore, the eigenspace associated with $\lambda_n(\varepsilon)$ is of dimension
$1$. 
\end{lemma}
\begin{proof}
  We follow the same procedure used in \cite{aletal13,aletal12} to
  prove existence of eigenvalues. Let $u_{n,N}$ be given by
  \eqref{eq:8} and set $\psi_{n,N}=\chi_\varepsilon u_{n,N}$. Let
  $\varepsilon^{\alpha}<r<2\varepsilon^\alpha$ be such that $\partial B(\Lambda_{n,N},r)\in\rho(\A_\varepsilon)$. Let
  further $\lambda\in\partial B(\Lambda_{n,N},r)$. Then, by \eqref{eq:7} we have
  \begin{displaymath}
    (\A_\varepsilon-\lambda)\psi_{n,N}= (\Lambda_{n,N}-\lambda)\psi_{n,N}+\varepsilon^{N+1}f\,,
  \end{displaymath}
where $\|f\|_2\leq C$,  for some $C>0$ which is independent of
$\varepsilon$. Applying $(\A_\varepsilon-\lambda)^{-1}$ to both sides of the above equation
yields
\begin{displaymath}
  (\A_\varepsilon-\lambda)^{-1}\psi_{n,N}=\frac{1}{\Lambda_{n,N}-\lambda}\big[\psi_{n,N}-\varepsilon^{N+1}(\A_\varepsilon-\lambda)^{-1}f\big] \,.
\end{displaymath}
Integrating the above identity with respect to $\lambda$ along
$\partial B(\Lambda_{n,N},r)$ yields
\begin{displaymath}
  P_n\psi_{n,N}=
  \psi_{n,N}-\oint_{\partial B(\Lambda_{n,N},r)}\frac{\varepsilon^{N+1}(\A_\varepsilon-\lambda)^{-1}f}{2\pi i(\Lambda_{n,N}-\lambda)}\,d\lambda \,,
\end{displaymath}
where $P_n$ is the spectral projection
\begin{equation}
\label{eq:21}
  P_n =\frac{1}{2\pi i}\oint_{\partial B(\Lambda_{n,N},r)}(\A_\varepsilon-\lambda)^{-1}\,d\lambda \,.
\end{equation}
With the aid of \eqref{eq:9} we then obtain that
\begin{equation}
  \label{eq:22}
\|(I-P_n)\psi_{n,N}\|_2 \leq C\varepsilon^{N+1-\alpha} \,.
\end{equation}
By Cauchy Theorem we now readily obtain that
\begin{displaymath}
  \sigma(\A_\varepsilon)\cap B(\Lambda_{n,1},2\varepsilon^\alpha)\neq \emptyset \,.
\end{displaymath}

We now prove that $P_nL^2(0,a/\varepsilon)$ is one dimensional. To this end
suppose that for some $\nu_1,\nu_2\in B(\Lambda_{n,1},2\varepsilon^\alpha)$ (which can be
equal or not) and $w_1,w_2\in D(\A_\varepsilon)$ we have 
\begin{equation}
\label{eq:23}
    (\A_\varepsilon-\nu_j)w_j=0 \quad j=1,2
\end{equation}
such that $\|w_1\|_2=\|w_2\|_2=1$ and
\begin{equation}
  \label{eq:24}
\langle\bar{w}_1,w_2\rangle=0\,.
\end{equation}
Let further
\begin{equation}
\label{eq:25}
  f_j= (A_0-\Lambda_{n,0})(\chi_\varepsilon w_j) \quad j=1,2\,.
\end{equation}
A simple calculation yields
\begin{equation}
\label{eq:26}
  f_j= \chi_\varepsilon(\nu_j-\Lambda_{n,0})w_j - i(V_\varepsilon-\varepsilon^{-1}V(0)-\beta_0x)\chi_\varepsilon w_j
  +[A_0,\chi_\varepsilon]w_j\quad j=1,2\,.
\end{equation}

We now turn to estimate the various terms on the right-hand-side of
\eqref{eq:26}. Let $j\in\{1,2\}$. For the first term we easily obtain, since
$\nu_j\in B(\Lambda_{n,1},2\varepsilon^\alpha)$ that
\begin{equation}
\label{eq:27}
  \|\chi_\varepsilon(\nu_j-\Lambda_{n,0})w_j\|_2\leq C\varepsilon\,.
\end{equation}
For the second term we have that
\begin{equation}
\label{eq:28}
  \|(V_\varepsilon-\varepsilon^{-1}V(0)-\beta_0x)\chi_\varepsilon w_j\|_2\leq C\varepsilon^{1-2\rho} \,.
\end{equation}
To estimate the last term we take the inner product of  \eqref{eq:23}
with $w_j$ to obtain from the real part that
\begin{displaymath}
   \|\nabla w_j\|_2\leq C \,.
\end{displaymath}
Consequently, we have that
\begin{displaymath}
\|[A_0,\chi_\varepsilon]w_j\|_2\leq \|\Delta\chi_\varepsilon w_j\|_2+2\|\nabla\chi_\varepsilon\cdot\nabla w_j\|_2\leq
C\varepsilon^{1-\rho} \,.
\end{displaymath}
Substituting the above, together with \eqref{eq:27} and \eqref{eq:28}
into \eqref{eq:26} then yields
\begin{equation}
\label{eq:29}
  \|f_j\|_2\leq C\varepsilon^{1-2\rho} \,.
\end{equation}

We now write
\begin{displaymath}
  \chi_\varepsilon w_j =  (\chi_\varepsilon w_j)_\|+ (\chi_\varepsilon w_j)_\perp \,,
\end{displaymath}
where
\begin{displaymath}
  (\chi_\varepsilon w_j)_\| =\langle\bar{u}_0,\chi_\varepsilon w_j\rangle u_0 \,.
\end{displaymath}
Applying Riesz-Schauder theory to $A_0$ yields, by \eqref{eq:25} and
\eqref{eq:26},
\begin{displaymath}
  \|(\chi_\varepsilon w_j)_\perp\|\leq C\varepsilon^{1-2\rho} \,.
\end{displaymath}
Consequently,
\begin{displaymath}
  |\langle\chi_\varepsilon{\bar w}_1,\chi_\varepsilon w_2\rangle| \geq 1-C\varepsilon^{1-2\rho} \,.
\end{displaymath}
Hence, by \eqref{eq:24} we have that
\begin{equation}
\label{eq:30}
  |\langle\tilde{\chi}_\varepsilon{\bar w}_1,\tilde{\chi}_\varepsilon w_2\rangle| \geq 1-C\varepsilon^{1-2\rho} \,.
\end{equation}

To complete the proof we take again the inner product of  \eqref{eq:23}
with $w_j$ to obtain, this time from the imaginary  part, that
\begin{displaymath}
   \|(V_\varepsilon-\varepsilon^{-1}V(0))w_j\|_2\leq C \,.
\end{displaymath}
Hence,
\begin{displaymath}
  \|w_j\|_{L^2(\varepsilon^{\rho-1},a/\varepsilon)}\leq C\varepsilon^{1-\rho}\,,
\end{displaymath}
from which we easily conclude that
\begin{displaymath}
   |\langle\tilde{\chi}_\varepsilon{\bar w}_1,\tilde{\chi}_\varepsilon w_2\rangle| \leq
   \|w_1\|_{L^2(\varepsilon^{\rho-1},a/\varepsilon)} \|w_2\|_{L^2(\varepsilon^{\rho-1},a/\varepsilon)}\leq C\varepsilon^{2(1-\rho)}\,,
\end{displaymath}
contradicting \eqref{eq:30} and therefore \eqref{eq:24}.
\end{proof}

\begin{proof}[Proof of Theorem  \ref{thm:1D} ]
Recall that by \eqref{eq:7} we have
  \begin{displaymath}
    (\A_\varepsilon-\Lambda_{n,N})\psi_{n,N}= \varepsilon^{N+1}f\,,
  \end{displaymath}
  where $\|f\|_2$ is uniformly bounded as $\varepsilon\to0$.  We now apply the
  spectral projection $P_n$, defined in \eqref{eq:21} to both side of
  the above equations. It can be easily verified that $[P_n,\A_\varepsilon]=0$.
  Consequently
\begin{equation}
\label{eq:31}
  (\A_\varepsilon-\Lambda_{n,N})P_n\psi_{n,N}= \varepsilon^{N+1}P_nf\,.
\end{equation}
By \eqref{eq:20} we have that
\begin{equation}
\label{eq:32}
  (\A_\varepsilon-\Lambda_{n,N})P_n\psi_{n,N}=(\lambda_n-\Lambda_{n,N})P_n\psi_{n,N}\,.
\end{equation}
By \eqref{eq:22}  we have that
\begin{displaymath}
  \|P_n\psi_{n,N}\|_2\geq 1-C\varepsilon^{N+1} \,.
\end{displaymath}
Substituting the above, together with \eqref{eq:32} into \eqref{eq:31}
then yields
\begin{displaymath}
  |\lambda_n-\Lambda_{n,N}|\leq C\varepsilon^{N+1}
\end{displaymath}
Theorem \ref{thm:1D} now easily follows from \eqref{eq:6}
\end{proof}

\section{Two dimensions: Quasimode construction} 
\label{sec:2}
Let $\Omega\subset\subset\R^2$ be a $C^3$ domain and $V\in C^3(\bar{\Omega})$.  Let
$\partial\Omega_\perp$ denote the portion of the boundary $\partial\Omega$ where $\nabla V$ is
orthogonal to $\partial\Omega$.  (Note that $\partial\Omega_\perp$ may be finite, but is
never empty by the continuity of $V$ on $\partial\Omega$.)  Let $x_0\in\partial\Omega_\perp$
such that
\begin{displaymath}
 |\nabla V(x_0)| = \min_{x\in\partial \Omega_\perp}|\nabla V(x)|\,,
\end{displaymath}
and let $V_0 = V(x_0)\,$. We look for an approximation of the principal eigenvalue and the
corresponding eigenfunction of the operator
\begin{equation}
\label{eq:33}
  \A_h = -h^2\Delta + i(V-V_0) \,,
\end{equation}
defined over
\begin{displaymath}
  D(\A_h) = H^1_0(\Omega,\C)\cap H^2(\Omega,\C)\,.
\end{displaymath}

Define in a vicinity of $\partial\Omega$ a curvilinear coordinate system $(t,s)$
such that $t=d(x,\partial\Omega)$ and $s(x)$ denotes the distance (or arclength)
along $\partial\Omega$ connecting $x_0$ and the projection of $x$ on $\partial\Omega$. We
have
\begin{equation}
\label{eq:34}
  \Delta = \Big(\frac{1}{g} \frac{\partial}{\partial s}\Big)^2 + \frac{1}{g}
  \frac{\partial}{\partial t}\Big(g\frac{\partial}{\partial t}\Big) \,, 
\end{equation}
where
\begin{equation}
\label{eq:35}
  g= 1- t\kappa(s)\,,
\end{equation}
and $\kappa(s)$ is the curvature at $s$ on $\partial\Omega$.  Expanding $\Delta$ near
$x_0$ ($t^2+s^2\ll1$) yields for some $u\in D(\A_h)$
\begin{equation}
\label{eq:36}
  \Delta u = u_{tt}+ u_{ss}+ \Upsilon u \,,
\end{equation}
where
\begin{equation}
\label{eq:37}
  \Upsilon u =\Big(\frac{1}{g^2}-1\Big)u_{ss}+
 \frac{ t\kappa^\prime}{g^3}u_s-\frac{\kappa}{g}u_t \,.
\end{equation}
We next expand $V$ near $x_0$ 
\begin{equation}
\label{eq:38}
   V(s,t) -V_0 = ct + \frac{1}{2}(\alpha s^2
  +\beta t^2+2\sigma st) + \OO((s^2+t^2)^{3/2}) \,,  
\end{equation}
where
\begin{displaymath}
  c = V_t(x_0) \quad ; \quad \alpha = V_{ss}(x_0)  \quad ; \quad
  \beta = V_{tt}(x_0) \quad;\quad  \sigma = V_{st}(x_0)  \,.
\end{displaymath}
We note that $V_s(x_0)=0$ since $x_0\in\partial\Omega_\perp$. We confine the
discussion, in view of \eqref{eq:3} to the case where $\alpha c>0$. Without
any loss of generality we may assume $c>0$ (and hence $\alpha>0$ as well),
otherwise we can consider the spectrum of the complex conjugate of $\A_h$.

We search for an approximate eigenpair $(u,\lambda)$ of $\A_h$. Previous
works \cite{al08,hen15} suggest that one should look for such $u$
which is localized near $x_0$. Applying the transformation
\begin{equation}
\label{eq:39}
  \tau= \Big(\frac{c}{h^2}\Big)^{1/3}t \quad;\quad \xi=\Big(\frac{\alpha}{h^2}\Big)^{1/4}s
\end{equation}
to \eqref{eq:38} and \eqref{eq:36}  leads to the
following approximation for every $u\in D(\A_h)$
\begin{equation}
 \label{eq:40}
 \frac{\alpha}{\varepsilon c^2} \A_h u =   -u_{\tau\tau} + i\tau u + \varepsilon^{1/2}\Big(-u_{\xi\xi}+\frac{i}{2}\xi^2u\Big)
+ \Big(\frac{\varepsilon}{\alpha}\Big)^{3/4}i\sigma\xi\tau u + Ru \,,
\end{equation}
where 
\begin{equation}
\label{eq:137}
  \varepsilon=\alpha(h^2/c^4)^{1/3}\,, 
\end{equation}
$\|u\|_2=1$, and the operator $R$ satisfies, for all $u\in D(\A_h)\,$
\begin{multline}
  Ru=
  c^{2/3}\Big(\frac{\varepsilon}{\alpha}\Big)^{1/2}\Big(\frac{1}{g^2}-1\Big)u_{\xi\xi}+
  c^{2/3}\Big(\frac{\varepsilon}{\alpha}\Big)^{9/4}\frac{
    \tau c^{1/3}\kappa^\prime}{g^3}u_\xi-\Big(\frac{\varepsilon}{\alpha}\Big)\frac{c^{1/3}\kappa}{g}u_\tau+  \\
  i\frac{\alpha}{\varepsilon c^2}\Big(V(\xi,\tau)-V_0-\frac{\varepsilon}{\alpha} c^2\tau-
  \frac{c^2\varepsilon^{3/2}}{\alpha}\frac{1}{2}\xi^2-
  \Big(\frac{\varepsilon}{\alpha}\Big)^{7/4}c^2\sigma\xi\tau\Big) \,. \label{eq:R}
\end{multline}
It can be easily verified that for any $0<\gamma <1$ we have 
\begin{multline}
 \|Ru\|_{L^2(B_+(0,\varepsilon^{-\gamma}))}\leq C\varepsilon\Big[\|\varepsilon^{1/2}|\tau u_{\xi\xi}| +\varepsilon^{5/4}|\tau u_\xi|+|u_\tau|
\|_{L^2(B_+(0,\varepsilon^{-\gamma}))} + \\
 C\varepsilon\Big[\|\tau^2u\|_{L^2(B_+(0,\varepsilon^{-\gamma}))}+\varepsilon^{1/4}\|\xi^3u\|_{L^2(B_+(0,\varepsilon^{-\gamma}))}\Big]\,.\label{eq:estR}
\end{multline}

We seek an approximate solution for $\A_h u=\lambda u$. To this end, we
introduce the expansion
\begin{displaymath}
  u\cong u_0+\varepsilon^{1/4}u_1 +\varepsilon^{1/2}u_2 + \varepsilon^{3/4}u_3 + \OO(\varepsilon) \quad ; \quad
  \frac{\alpha}{\varepsilon c^2} \lambda=\lambda_0+\varepsilon^{1/4}\lambda_1 + \varepsilon^{1/2}\lambda_2+ \varepsilon^{3/4}\lambda_3+ \OO(\varepsilon) \,.
\end{displaymath}
Substituting into \eqref{eq:40} leads to the following $\OO(1)$ balance
\begin{subequations}
\label{eq:41}
  \begin{equation}
 \LL_\tau u_0 \overset{def}{=} - \frac{\partial^2u_0}{\partial\tau^2}+ i\tau u_0 = \lambda_0u_0 \quad;\quad u_0(0,\xi)=0 \,,
\end{equation}
where the operator $\LL_\tau$ is defined over
\begin{equation}
  D(\LL_\tau) = \{ u\in H^2(\R_+,\C)\cap H^1_0(\R_+,\C) \, | \, \tau u\in L^2(\R,\C)\}\,.
\end{equation}
\end{subequations}
The solution to \eqref{eq:41} associated with the energy $\lambda_0$ having the smallest real
part is given by
\begin{equation}
\label{eq:42}
  u_0(\tau,\xi)=v_0(\tau)w_0(\xi) \quad \text{where} \quad  v_0(\tau)= A_i(e^{i\pi/6}\tau+\mu_1)\,,
\end{equation}
and
\begin{equation}
\label{eq:43}
  \lambda_0 = e^{-i2\pi/3}\mu_1 \,,
\end{equation}
where $A_i$ is Airy's function and $\mu_1<0$ is its rightmost zero. The
function $w_0(\xi)$ will be determined from the $\OO(\varepsilon^{1/2})$
balance. 

The next order, or $\OO(\varepsilon^{1/4})$, balance in \eqref{eq:40} assumes the
form
\begin{equation}
\label{eq:44}
  (\LL_\tau-\lambda_0) u_1 = \lambda_1u_0 \quad;\quad u_1(0,\xi)=0 \,,
\end{equation}
Taking the inner product of \eqref{eq:44} with $\bar{v}_0$ yields
$\lambda_1=0$. Hence, $u_1=v_0(\tau)w_1(\xi)$.

The next order, or $\OO(\varepsilon^{1/2})$, balance in \eqref{eq:40} assumes the
form
\begin{equation}
\label{eq:45}
  (\LL_\tau-\lambda_0) u_2 =- \Big(\LL_\xi -\lambda_2\Big)u_0 \quad;\quad u_2(0,\xi)=0 \,,
\end{equation}
where
\begin{equation}
\label{eq:46}
 \LL_\xi = -\frac{\partial^2}{\partial\xi^2} +\frac{i}{2}\xi^2\,,
\end{equation}
is defined over
\begin{displaymath}
  D(\LL_\xi) = \{ u\in H^2(\R,\C) \, | \, \xi^2u\in L^2(\R,\C)\}
\end{displaymath}
For fixed $\xi$ we now take the inner product of the above equation
with $\bar{v}_0$, in $L^2(\R_+)$. After noticing that by Cauchy's Theorem
\begin{equation}\label{intv0}
  \int_0^\infty v_0^2(\tau) \,d\tau = e^{-i\pi/6}\int_0^\infty A_i^2(x+\mu_1)
\,dx \neq 0\,,
\end{equation}
we obtain
\begin{displaymath}
  (\LL_\xi -\lambda_2)w_0=0 \,.
\end{displaymath}
The solution  of the above problem corresponding to the $\lambda_2$ with smallest real part is given by 
\begin{equation}
\label{eq:47}
  w_0(\xi) = C_0\exp\Big\{-\frac{1}{\sqrt{2}}e^{i \frac{\pi}{4}}  
  \xi^2\Big\} \quad ; \quad \lambda_2 = \sqrt{2}e^{i \frac{\pi}{4}}  \,.
\end{equation}
The constant $C_0$ should be obtain, up to a product by $-1$, from the
normalization condition $\|u\|_2=1$. We allow dependence of $C_0$ on
$\varepsilon$ (see below).  Substituting into \eqref{eq:45} yields
\begin{displaymath}
  u_2(\tau,\xi)=v_0(\tau)w_2(\xi) \,.
\end{displaymath}

For the $\OO(\varepsilon^{3/4})$ balance in \eqref{eq:40} we have
\begin{displaymath}
  (\LL_\tau -\lambda_0)u_3 = -v_0(\LL_\xi-\lambda_2)w_1 -
  \Big(i\sigma\xi\tau-\lambda_3\Big)v_0w_0 \quad ; \quad u_2(0,\xi)=0\,.
\end{displaymath}
We take once again the inner product of the above balance with
$\bar{v}_0$ to obtain
\begin{equation}
\label{eq:48}
  (\LL_\xi-\lambda_2)w_1+  \Big(i\gamma\xi-\lambda_3\Big)w_0 =0 \,,
\end{equation}
where
\begin{displaymath}
  \gamma=\sigma\frac{\int_0^\infty \tau v_0^2(\tau) \,d\tau}{\int_0^\infty v_0^2(\tau) \,d\tau} \,.
\end{displaymath}
Note that this expression is well-defined due to \eqref{intv0}.
Taking the inner product, this time in $L^2(\R,\C)$, of \eqref{eq:48}
with $w_0$, which is even,  yields
\begin{displaymath}
  \lambda_3 = 0\,.
\end{displaymath}
Furthermore, $w_1$ is the unique solution of 
\begin{displaymath}
   (\LL_\xi-\lambda_2)w_1=  -i\gamma\xi w_0 \quad ;\quad \int_\R w_1(\xi)w_0(\xi)\,d\xi=0\,,
\end{displaymath}
and
\begin{displaymath}
  u_3 = v_3(\xi,\tau)+ v_0(\tau)w_3(\xi) \,,
\end{displaymath}
where $v_3$ is the unique solution of the problem
\begin{equation}
  \begin{cases}
    (\LL_\tau -\lambda_0)v_3 = -i\xi(\tau-\gamma)v_0w_0 & \tau>0 \label{eq:v2}\\
    v_3(0,\xi)=0 & \\
    \int_0^\infty v_2(\tau,\xi)v_0(\tau) d\tau = 0\,. 
  \end{cases}
\end{equation}
Notice that, if $\mathcal{S}(\mathbb{R}_+^2)$ denotes the Schwartz
space of rapidly decaying functions along with all their derivatives,
then the right-hand side in \eqref{eq:v2} belongs to
$\mathcal{S}(\mathbb{R}_+^2)$. As the operator
$-\partial^2/\partial\tau^2+i\tau-\lambda_0$ is globally elliptic with respect
to $\tau$, in the sense of \cite[Definition 1.5.6]{hel84},
we have that
\begin{equation}\label{eq:v2inS}
 v_3\in\mathcal{S}(\mathbb{R}_+^2)\,,
\end{equation}
(see \cite[Theorem 1.6.4]{hel84}).
For the same reason, the $\OO(\varepsilon)$ balance would
yield $w_3\in\mathcal{S}(\mathbb{R})$.

We have thus obtained the quasimode
\begin{multline}
\label{eq:49}
U=  \Big(C_0(\varepsilon)\exp\Big\{-\frac{1}{\sqrt{2}}e^{i\frac{\pi}{4}}  
  \xi^2\Big\}+\varepsilon^{1/2}w_1(\xi) \Big)A_i(e^{i\pi/6}\tau+\mu_1)\\ + \varepsilon^{3/4}v_3(\xi,\tau) + \varepsilon^{3/4} w_3(\xi)A_i(e^{i\pi/6}\tau+\mu_1) \,.
\end{multline}
We obtain the various constants by requiring that
\begin{displaymath}
  \|U\|_2=1 \,.
\end{displaymath}

We now conclude this section by the following proposition
\begin{proposition}
  Let $\A_h$ be given by \eqref{eq:33} and $U$ by \eqref{eq:49}. Let
  further
  \begin{equation}
\label{eq:50}
    \Lambda= \lambda_0+\varepsilon^{1/2}\lambda_2 \,.
  \end{equation}
Let $\eta_r = \eta_r^0(\tau)\eta_r^1(\xi)$, where  $\eta^0_r\in C^\infty(\R_+,[0,1])$ and
$\eta_r^1\in C^\infty(\R,[0,1])$ are chosen so that
\begin{equation}
\label{eq:51}
\eta_r =
  \begin{cases}
   1 & |x-x_0|<r \\
   0 & |x-x_0|>2r \,,
  \end{cases}
|\nabla\eta_r|\leq \frac{C}{r}\,.
\end{equation} 
Then,
\begin{equation}
\label{eq:52}
  \Big\|\Big(\frac{\alpha}{\varepsilon c^2} \A_h-\Lambda\Big)(\eta_{\varepsilon^{-1/2}}U)\Big\|_2 \leq C\varepsilon\|\eta_{\varepsilon^{-1/2}}U\|_2 \,. 
\end{equation}
\end{proposition}
\begin{proof}
We first write
\begin{eqnarray}
   \frac{\alpha}{\varepsilon c^2}\A_h(\eta_{\varepsilon^{-1/2}}U) & = &
 \big(\LL_\tau+\varepsilon^{1/2}\LL_\xi+\varepsilon^{3/4}i\sigma\xi\tau\big)(\eta_{\varepsilon^{-1/2}}U)
  + R\eta_{\varepsilon^{-1/2}}U \nonumber \\
  & = & \Lambda\eta_{\varepsilon^{-1/2}}U + \Big[\LL_\tau+\varepsilon^{1/2}\LL_\xi , \eta_{\varepsilon^{-1/2}}\Big] U
  + R\eta_{\varepsilon^{-1/2}}U\,, \label{eq:cutoff1}
\end{eqnarray}
where the operator $R$ is defined by \eqref{eq:R}. We next seek an
estimate for the commutator term in \eqref{eq:cutoff1}, given by
\begin{equation}\label{eq:cutoff2}
 [\LL_\tau,\eta_{\varepsilon^{-1/2}}]U = -\partial_\tau^2(\eta_{\varepsilon^{-1/2}})U-2\partial_\tau\eta_{\varepsilon^{-1/2}}\partial_\tau U\,.
\end{equation}
In order to estimate the norm of $U$ and $\partial_\tau U$ on the support of
$\partial_\tau^2\eta_{\varepsilon^{-1/2}}$ and $\partial_\tau\eta_{\varepsilon^{-1/2}}\,$, we recall the well-known
asymptotic behavior of the Airy function \cite{abst72}:
\begin{equation}\label{eq:cutoff3}
 Ai(z) = \frac{1}{2\sqrt{\pi}z^{1/4}}e^{-\frac{2}{3}z^{3/2}}\big(1+\OO(z^{-3/2})\big)
\end{equation}
as $|z|\to+\infty$ in any sector of the form $|\arg z|\leq\pi-\delta\,$, $\delta>0\,$. By
\eqref{eq:49}, and since for all $(\tau,\xi)\in\textrm{Supp }\pa_\tau\eta_{\varepsilon^{-1/2}}$
we have $\varepsilon^{-1/2}\leq\tau\leq2\varepsilon^{-1/2}\,$, \eqref{eq:v2inS} and
\eqref{eq:cutoff3} yield
\begin{displaymath}
  \|(\partial_\tau^2\eta_{\varepsilon^{-1/2}}) U\|_2\leq C_1\varepsilon\,,
\end{displaymath}
for some positive constant $C_1\,$. 

Since the asymptotic behaviour of
$Ai'$, as $|z|\to\infty$ is not substantially different from
\eqref{eq:cutoff3} (cf. \cite{abst72}), we easily obtain that
\begin{displaymath}
 \|\partial_\tau\eta_{\varepsilon^{-1/2}}\partial_\tau U\|_2\leq C_2\varepsilon\,,~~~~C_2>0\,.
\end{displaymath}
Thus \eqref{eq:cutoff2} yields, for some $C>0\,$,
\begin{equation}\label{eq:cutoff4}
 \|[\LL_\tau,\eta_{\varepsilon^{-1/2}}]U\|_2\leq C\varepsilon\,.
\end{equation}
Due to the decay of the $U$ and $\pa_\xi U$ as $|\xi|\to+\infty$ (recall that
$w_3\in\mathcal{S}(\mathbb{R})$), we similarly obtain
\begin{equation}\label{eq:cutoff5}
 \|[\varepsilon^{1/2}\LL_\xi,\eta_{\varepsilon^{-1/2}}]U\|_2\leq K\varepsilon\,,
\end{equation}
for some $K>0\,$.can be estimated as
follows. Using

To estimate the remaining term $R\eta_{\varepsilon^{-1/2}}U$ we use \eqref{eq:estR}
to obtain, for $\alpha\in(1/2,1)$,
\begin{equation}\label{eq:cutoff6}
 \|R\eta_{\varepsilon^{-1/2}}U\|_2\leq \|RU\|_{L^2(B_+(0,\varepsilon^{-\alpha}))}\leq C'\varepsilon
\end{equation}
for some $C'>0$.\ Finally \eqref{eq:cutoff1}, \eqref{eq:cutoff4},
\eqref{eq:cutoff5} and \eqref{eq:cutoff6} yield, for some positive
$\tilde{C}$ and $C$,
\begin{eqnarray*}
 \Big\|\Big(\frac{\alpha}{\varepsilon c^2} \A_h-\Lambda\Big)(\eta_{\varepsilon^{-1/2}}U)\Big\|_2 & \leq & C'\varepsilon \\
 & \leq & C\varepsilon\|\eta_{\varepsilon^{-1/2}}U\|_2\,,
\end{eqnarray*}
where we have used the that for some $C''>0$,
$\|\eta_{\varepsilon^{-1/2}}U\|_{2}\geq1/C''$. 
\end{proof}

\section{Eigenvalue existence }
\label{sec:3}
Let $\LL_\tau$ and $\LL_\xi$ be respectively defined by \eqref{eq:41}
and \eqref{eq:46}. Then let
\begin{equation}
\label{eq:68}
  \B_\varepsilon  = \LL_\tau + \varepsilon^{1/2}\LL_\xi 
\end{equation}
be the closed operator associated with the quadratic form
\begin{displaymath}
  \langle\nabla u,\nabla v\rangle+ i\langle u, (\tau+ \varepsilon^{1/2}\xi^2)v\rangle 
\end{displaymath}
whose domain is given by $\tilde{V}\times  \tilde{V}$ where
\begin{displaymath}
  \tilde{V}=\{ u\in H^1_0(\R^2_+,\C) \,| \,
  |(\tau^{1/2}+|\xi|)u\in L^2(\R^2_+,\C)\} \,. 
\end{displaymath}
It can be easily verified that
\begin{displaymath}
  D(\B_\varepsilon) = \{u\in H^2(\R_+^2,\C)\cap H^1_0(\R^2_+)\,|\, (\tau+\xi^2)u\in L^2(\mathbb{R}_+^2),\} \,.
\end{displaymath}
We begin by the following straightforward observation
\begin{lemma}
  We have 
  \begin{equation}
    \label{eq:69}
\sigma(\B_\varepsilon)=\{c^{2/3}\mu_ne^{-i2\pi/3}+(2k-1)\varepsilon^{1/2}\sqrt{2}e^{i\frac{\pi}{4}}\}_{n,k=1}^\infty \,.
  \end{equation}
\end{lemma}
\begin{proof}
After the scale changes $\tau\mapsto c^{1/3}\tau$ and
$\xi\mapsto\big(|\alpha|/2\big)^{1/4}\xi$, we obtain from \cite{al08} and \cite[Section $14.5$]{da07} the following expressions for the eigenvalues of the complex Airy
operator $\LL_\tau$ and the complex harmonic oscillator $\LL_\xi$:
\begin{displaymath}
 \sigma(\LL_\tau) = \Big\{c^{2/3}\mu_ne^{-i2\pi/3} : n\geq1\Big\}\,,
\end{displaymath}
$\mu_n$ being the $n$-th (negative) zero of the Airy function $Ai\,$, and
\begin{displaymath}
 \sigma(\LL_\xi) = \Big\{(2k-1)\sqrt{2}\,e^{i\frac{\pi}{4}} : k\geq1\Big\}\,.
\end{displaymath}

Denote by $\LL_\tau\dotplus\varepsilon^{1/2}\LL_\xi$ the closure of the operator
$\LL_\tau\otimes I+I\otimes(\varepsilon^{1/2}\LL_\xi)$ whose domain is $D(\LL_\tau)\otimes D(\LL_\xi)$.  We
first need to verify that the domains of $\B_\varepsilon$ and
$\LL_\tau\dotplus\varepsilon^{1/2}\LL_\xi$ coincide.  Let $e^{-t\B_\varepsilon}$ denote the
contraction semigroup generated by $\B_\varepsilon$, and let $\varphi\in D(\LL_\tau)$, $\psi\in
D(\LL_\xi)$. Clearly,
\begin{displaymath}
 e^{-t\B_\varepsilon}(\varphi\otimes\psi) = e^{-t\LL_\tau}\varphi\otimes e^{-t(\varepsilon^{1/2}\LL_\xi)}\psi\,,
\end{displaymath}
where $e^{-t\LL_\tau}$ and $e^{-t(\varepsilon^{1/2}\LL_\xi)}$ denote respectively the
contraction semigroups generated by $\LL_\tau$ and $\varepsilon^{1/2}\LL_\xi$.Thus, 
\begin{displaymath}
 e^{-t\B_\varepsilon}\big(D(\LL_\tau)\otimes D(\LL_\xi)\big)\subset D(\LL_\tau)\otimes D(\LL_\xi)\,.
\end{displaymath}
Consequently, due to \cite[Theorem X.49]{resi78} we have $\B_\varepsilon =
\overline{{(\B_\varepsilon)}_{\mid D(\LL_\tau)\otimes D(\LL_\xi)}}\,$, and $\B_\varepsilon$ clearly
coincides with $\LL_\tau\otimes I+I\otimes(\varepsilon^{1/2}\LL_\xi)$ on $D(\LL_\tau)\otimes D(\LL_\xi)\,$,
and hence $\B_\varepsilon = \LL_\tau\dotplus\varepsilon^{1/2}\LL_\xi\,$.

Noticing that $\LL_\tau$ and $\LL_\xi$ are both sectorial with respect to
the same sector $\mathcal{S} = \{z\in\mathbb{C} : 0\leq\arg z\leq\pi/2\}\,$, we
can then apply the so-called Ichinose Lemma (see \cite[Theorem
XIII.35, Corollary 2]{resi78}) which yields
\begin{displaymath}
 \sigma\big(\LL_\tau\dotplus\varepsilon^{1/2}\LL_\xi\big) = \sigma(\LL_\tau) + \sigma(\varepsilon^{1/2}\LL_\xi)\,,
\end{displaymath}
and \eqref{eq:69} follows.
\end{proof}

The following auxiliary lemma will be necessary in the sequel
\begin{lemma}
  Let $v_n$ denote the (unique up to multiplication by a complex number of modulus $1$)
  unity norm eigenfunction associated with the eigenvalue
  \begin{equation}
\label{eq:55}
    \nu_{n-1}=\mu_ne^{-i2\pi/3} \quad n\in\N
  \end{equation}
of $\LL_\tau$. Let further $\Vg$ denote the
  form domain of $\LL_\tau$, i.e,
  \begin{displaymath}
    \Vg = \{u\in H^1_0(\R_+,\C)\, | \, \tau^{1/2}u\in L^2(\R_+,\C) \,\}\,,
  \end{displaymath}
and $\Vg_n={\rm span} \{v_n\}_{n=k+1}^\infty \cap\Vg$. Set
  \begin{subequations}
    \label{eq:70}
     \begin{equation}
\beta_k = \inf_{
  \begin{subarray}{c}
   u\in\Vg_n \\
   \|u\|=1
  \end{subarray}} \|u_\tau\|_2^2+\|\tau^{1/2}u\|_2^2 \,.
  \end{equation}
Then,
\begin{equation}
   \beta_k\to\infty \,.
\end{equation}
  \end{subequations}
\end{lemma}
\begin{proof}
  Let us assume by contradiction that there exists a subsequence
  $(k_n)$ and a positive constant $C$ such that
\begin{displaymath}
 \sup_{n\in\mathbb{N}}\beta_{k_n}\leq C\,.
\end{displaymath}
Then there exists a sequence $(u_n)$ of functions in $H^1_0(\R_+,\C)$,
$\tau^{1/2}u_n\in L^2(\R_+,\C)$ such that, for all $n\in\mathbb{N}$, 
$u_n\in{\rm span} \{v_j\}_{j=k_n+1}^\infty$, $\|u_n\|_2 = 1$ and
\begin{equation}
\label{eq:71}
 \sup_{n\in\mathbb{N}}\big(\|\pa_\tau u_n\|_2^2+\|\tau^{1/2}u_n\|_2^2\big)\leq2C\, .
\end{equation}
Since for any $r>0$ we have
\begin{displaymath}
  \int_r^\infty |u_n|^2 \leq \frac{1}{r} \int_r^\infty \tau|u_n|^2 \leq \frac{2C}{r}\,, 
\end{displaymath}
we can choose such $r$ for which
\begin{displaymath}
   \int_0^r |u_n|^2 \geq \frac{1}{2} \,.
\end{displaymath}
Since by \eqref{eq:71} the $H^1(\R_+,\C)$ norms of $\{u_n\}_{n=1}^\infty$
are bounded, we can extract a subsequence $(u_{\varphi(n)})$ such that
$u_{\varphi(n)}\to u_\infty$ in $L^2(\mathbb{R}_+,\C)$ weakly, and in
$L^2([0,r],\C)$ strongly, for some limit function $u_\infty\in
L^2(\mathbb{R}_+,\C)$. We note that 
\begin{equation}
\label{eq:72}
   \int_0^r |u_\infty|^2 \geq \frac{1}{2} \,.
\end{equation}

Now let $k\in\mathbb{N}$ be fixed. Then for all $n$ such that $k_{\varphi(n)}\geq k$ we have
\begin{displaymath}
 u_{\varphi(n)} \in {\rm span} \{v_j\}_{j\geq k+1} = \big({\rm  span} \{\bar{v}_n\}_{n=1}^k\big)_\perp\,,
\end{displaymath}
hence, by the weak convergence in $L^2(\mathbb{R}_+,\C)$.
\begin{displaymath}
 0 = \langle u_{\varphi(n)} , \bar v_k\rangle\longrightarrow \langle u_\infty,\bar v_k\rangle = 0\,.
\end{displaymath}
Consequently $u_\infty\in\big({\rm span}
\{\bar{v}_j\}_{j=1}^{+\infty}\big)_\perp$, thus $u_\infty = 0$ since the
eigenfunctions $\{\bar v_j\}_{j\geq1}$ of $\LL_\tau^*$ form a complete
family of $L^2(\mathbb{R}_+,\C)$ (see \cite{al08}). A contradiction, in
view of \eqref{eq:72}.
\end{proof}

We next claim the following
\begin{lemma}
\label{lem:3.2}
There exist $r_0>0$, $\varepsilon_0>0$ and $C>0$, such that if $r\in(0,r_0)$, then
  \begin{equation}
    \label{eq:73}
|\lambda-\lambda_0-\varepsilon^{1/2}\lambda_2|=r\varepsilon^{1/2} \Rightarrow\|(\B_\varepsilon-\lambda)^{-1}\|\leq
\frac{C}{r}\varepsilon^{-1/2} \quad\forall0<\varepsilon<\varepsilon_0\,.
  \end{equation}
\end{lemma}
\begin{proof}
Suppose that $r$ is so chosen such that 
$\partial B(\lambda_0+\varepsilon^{1/2}\lambda_2,r\varepsilon^{1/2})\in\rho(\B_\varepsilon)$. Let $g\in {\rm
  span}\{v_nw_m\}_{n,m=0}^\infty$ and $w$ denote the solution of 
\begin{equation}
\label{eq:75}
  (\B_\varepsilon-\lambda)w=g \,.
\end{equation}
Let further
\begin{displaymath}
  \lambda-\lambda_0-\varepsilon^{1/2}\lambda_2=\varepsilon^{1/2}re^{i\alpha}\,,
\end{displaymath}
where $\alpha\in[0,2\pi)$. By the Riesz-Schauder Theory
(cf. \cite[Eq. (16.4)]{ag65} for instance) we
have that
\begin{equation}
  \label{eq:76}
(\LL_\tau-\lambda)^{-1} = \frac{\Pi_0}{\lambda-\nu_0} + \sum_{k=1}^K
\frac{\Pi_k}{\lambda-\nu_k} + T_k(\lambda)\,,
\end{equation}
where $\{\nu_n\}_{n=0}^\infty$ are given by \eqref{eq:55}, and $\|T_K\|\leq
C_K$ in $B(\nu_0,\tilde{r})$ for some fixed $\tilde{r}>0$.  In the
above $\Pi_k$ is the projection operator on ${\rm span}\{v_k\}$, which
is explicitly given, for any $u\in {\rm span}\{v_n\}_{n=0}^\infty$, by
\begin{displaymath}
  \Pi_k(u) = \langle\bar{v_k},u\rangle_\tau v_k(\tau)\,,
\end{displaymath}
where $\langle\cdot,\cdot\rangle_\tau$ denotes the standard  $L^2(\R_+,\C)$ inner
product. 

Let $u_k=\Pi_k(w)$. It can be easily verified that
\begin{displaymath}
  u_k = \varepsilon^{-1/2}(\LL_\xi - \lambda_2-re^{i\alpha}+\varepsilon^{-1/2}(\nu_k-\nu_0))^{-1}\Pi_k(g)\,.
\end{displaymath}
It easily follows from here that
\begin{equation}
\label{eq:77}
  \|u_0\|_2\leq  \frac{C}{r\varepsilon^{1/2}} \|\Pi_0(g)\|_2\leq \frac{C}{r\varepsilon^{1/2}}
  \|g\|_2\,,
\end{equation}
whereas
\begin{equation}
\label{eq:78}
 \|u_k\|_2\leq  C_k\|g\|_2\,,
\end{equation}
where $C_k$ is independent of $r$ and $\varepsilon$. For every $K\geq1$ we have
\begin{equation}
\label{eq:79}
  \|w\|_2 \leq \Big( \frac{C}{r\varepsilon^{1/2}} + \sum_{k=1}^KC_k\Big)\|g\|_2 +
  \|P_K(w)\|_2 \,,
\end{equation}
where
\begin{equation}
\label{eq:80}
  P_K =I - \sum_{k=0}^K\Pi_k \,.
\end{equation}

To complete the proof we need an estimate for  $\|P_K(w)\|_2$. Let
then $u_K=P_K(w)$. Clearly,
\begin{displaymath}
  (\B_\varepsilon-\lambda)u_K=P_K(g) \,.
\end{displaymath}
Taking the inner product of the above equation by $u_K$ yields
\begin{align*}
  & \Big\|\frac{\partial u_K}{\partial\tau}\Big\|_2^2 +
  \varepsilon^{1/2}\Big\|\frac{\partial u_K}{\partial\xi}\Big\|_2^2 - \Re\lambda\|u_K\|_2^2 =
  \Re\langle u_K,P_K(g)\rangle  \\
&\|\tau^{1/2}u_K\|_2^2 +
  \varepsilon^{1/2}\|\xi u_K\|_2^2 - \Im\lambda\|u_K\|_2^2 =
  \Im\langle u_K,P_K(g)\rangle  \,.
\end{align*}
Combining the above equations yields
\begin{equation}
\label{eq:81}
  \Big\|\frac{\partial u_K}{\partial\tau}\Big\|_2^2 + \|\tau^{1/2}u_K\|_2^2  -
  (\Im\lambda+\Re\lambda)\|u_K\|_2^2 \leq 2 \|u_K\|_2\|P_K(g)\|_2 \,.
\end{equation}
As
\begin{equation}
\label{eq:82}
  \|P_K(g)\|_2 \leq C_K\|g\|_2 \,,
\end{equation}
we obtain by \eqref{eq:70} and \eqref{eq:81} that for sufficiently
large $K$ (but independent of $\varepsilon$)
\begin{displaymath}
  \|u_K\|_2\leq C_K\|g\|_2 \,.
\end{displaymath}
 The lemma is now proved by the above and \eqref{eq:79} for any $g\in {\rm
  span}\{v_nw_m\}_{n,m=0}^\infty$, and hence for any $g\in L^2(\R^2_+,\C)$
via a density argument.
\end{proof}
Note that $r$ may depend on $\varepsilon$. As a matter of fact \eqref{eq:73}
remains valid independently of the pace at which $r\to0$ as $\varepsilon\to0$.

\begin{corollary}
  Under the conditions of \ref{lem:3.2} we have that
  \begin{equation}
    \label{eq:85}
\|(\B_\varepsilon-\lambda)^{-1}P_1\|\leq C \,,
  \end{equation}
where $C$ is independent of $\varepsilon$.
\end{corollary}
The corollary follows immediately from \eqref{eq:78} and
\eqref{eq:82}.\\

Recall now the definition of $\Sg$ from the introduction
\begin{displaymath}
  \Sg=\{x\in\partial\Omega_\perp\,: \,|\nabla V(x)| = |\nabla V(x_0)|  \,, \;V(x)=V(x_0)\}\,.
\end{displaymath}
By \eqref{eq:3}, $\Sg$ is a finite set of isolated points
$\{x_j\}_{j\in J_\Sg}$. Recall the definition of the curvilinear
coordinate system $(s,t)$ from the previous section, and then let
$x_j=(s_j,0)$. Let further $f\in L^\infty(\Omega,\C)$ be supported on
$\Omega\cap\Union_{j\in J_\Sg}B(x_j,\delta)$ and satisfy
\begin{equation}
\label{eq:90}
  |f|\leq C\|f\|_2\varepsilon^{7/8}e^{-\gamma_1\varepsilon^{-3/2}[(s-s_j)^2+t^{3/2}]} \quad
  \text{in}~ B(x_j,\delta)\cap\Omega \quad \forall j\in J_\Sg \,,
\end{equation}
for some fixed and positive $\gamma_1$ and $C$.

We seek an estimate for the resolvent of $\A_h$. To this end a few
auxiliary estimates, beyond \eqref{eq:73}, are necessary.  Set then
\begin{displaymath}
  \Omega_+=\{x\in\Omega\,|\, V(x)>V(x_0)\,\} \quad ; \quad \Omega_-=\{x\in\Omega\,|\,
  V(x)<V(x_0)\,\} \,,
\end{displaymath}
and
\begin{displaymath}
  \Gamma=\{x\in\Omega\,|\, V(x)=V(x_0)\,\} \,.
\end{displaymath}
Define then the cutoff function $ \chi_{\varepsilon,n}^+\in C^\infty(\Omega,[0,1])$,
where $n\in\N$, in the following manner
\begin{equation}
\label{eq:91}
  \chi_{\varepsilon,n}^+(x) =
  \begin{cases}
    1 & x\in\Omega_- \\
    1 & x\in\Omega_+\cap  \{V(x)-V(x_0)\leq2^{n-1}\varepsilon^\rho\} \\
    0 & x\in\Omega_+\cap \{V(x)-V(x_0)\geq 2^n\varepsilon^\rho\} \,,
  \end{cases}
\quad \|\nabla\chi_{\varepsilon,n}^+\|_\infty \leq \frac{C_n}{\varepsilon^\rho}
\end{equation}
where $0<\rho<1$. We further set
\begin{equation}
\label{eq:92}
  (\tilde{\chi}_{\varepsilon,n}^+)^2 +  (\chi_{\varepsilon,n}^+)^2=1 \,.
\end{equation}
In a similar manner we then define $ \chi_{\Gamma,\varepsilon,n}^-$: 
\begin{displaymath}
  \chi_{\varepsilon,n}^-(x) =
  \begin{cases}
    1 & x\in\Omega_+ \\
    1 & x\in\Omega_-\cap \{V(x_0)-V(x)\leq2^{n-1}\varepsilon^\rho\} \\
    0 & x\in\Omega_-\cap \{V(x_0)-V(x) \geq 2^n\varepsilon^\rho\} \,.
  \end{cases}
\end{displaymath}
The complementary cutoff function $\tilde{\chi}_{\varepsilon,n}^-$ is then given
by
\begin{displaymath}
  (\tilde{\chi}_{\varepsilon,n}^-)^2 = 1-  (\chi_{\varepsilon,n}^-)^2
\end{displaymath}

We begin with the following estimate
\\
\begin{lemma}
\label{lem:Gamma}
Let $f$ satisfy \eqref{eq:90} and 
\begin{equation}
\label{eq:93}
  (\A_h-\lambda^*)w=f\,,
\end{equation}
where
\begin{displaymath}
  |\lambda^*|\leq C\varepsilon
\end{displaymath}
Then, for any $n\in\N$ there exists $C_n>0$ and $\gamma_2>0$ such that for sufficiently
small $\varepsilon$ we have
\begin{subequations}
    \label{eq:94}
    \begin{equation}
  \| \tilde{\chi}_{\varepsilon,n}^-w\|_2 +\| \tilde{\chi}_{\varepsilon,n}^+w\|_2
  \leq C_n(\varepsilon^{n\rho-1}\|w\|_2+ e^{-\gamma_2\varepsilon^{-\frac{3}{2}(1-\rho)}}\|f\|_2)\,.
  \end{equation}
Furthermore, we have that
\begin{equation}
  \|\nabla(\tilde{\chi}_{\varepsilon,n}^+w)\|_2 +\| \nabla(\tilde{\chi}_{\varepsilon,n}^-w)\|_2 + \varepsilon^2(   \|D^2(\tilde{\chi}_{\varepsilon,n}^+w)\|_2 +\|
  D^2(\tilde{\chi}_{\varepsilon,n}^-w)\|_2) \leq  C_n\varepsilon^{n\rho-1}(\|w\|_2+\|f\|_2)\,.
\end{equation}
\end{subequations}
\end{lemma}
\begin{proof}
  In the following the constants $C$ and $\gamma_2$ depend on $n$.  Taking
  the inner product of (\ref{eq:93}) with $(\tilde{\chi}_{\varepsilon,n}^+)^2w$
  yields
\begin{subequations}
\label{eq:95}
\begin{empheq}[left={\empheqlbrace}]{alignat=2}
\|\nabla(\tilde{\chi}_{\varepsilon,n}^+w)\|_2^2 - \|w\nabla\tilde{\chi}_{\varepsilon,n}^+\|_2^2=
&\frac{\alpha}{\varepsilon^3c^4}\big(\Re\lambda^*\| \tilde{\chi}_{\varepsilon,n}^+w\|_2^2+ \Re\langle
\tilde{\chi}_{\varepsilon,n}^+w,\tilde{\chi}_{\varepsilon,n}^+f\rangle\big)& \\
 \frac{\alpha}{\varepsilon^3c^4}\|\tilde{\chi}_{\varepsilon,n}^+|V-V(x_0)|^{1/2}w\|_2^2&+
\Im\langle w\nabla\tilde{\chi}_{\varepsilon,n}^+,\nabla(\tilde{\chi}_{\varepsilon,n}^+w)\rangle&\notag \\
=&\frac{\alpha}{\varepsilon^3c^4}\big(\Im\lambda^*\| \tilde{\chi}_{\varepsilon,n}^+w\|_2^2+\Im\langle\tilde{\chi}_{\varepsilon,n}^+w,\tilde{\chi}_{\varepsilon,n}^+f\rangle\big)\,. &
\end{empheq}
\end{subequations}
From the definition of $\tilde{\chi}_{\varepsilon,n}^+$ and (\ref{eq:95}b) we get 
\begin{equation}
\label{eq:96}
  \| \tilde{\chi}_{\varepsilon,n}^+w\|_2^2\leq C\varepsilon^{3-\rho}
  \Big(\|\nabla(\tilde{\chi}_{\varepsilon,n}^+w)\|_2^2
  +\|w\nabla\tilde{\chi}_{\varepsilon,n}^+\|_2^2+ \varepsilon^{-4}\|\tilde{\chi}_{\varepsilon,n}^+f\|_2^2+\varepsilon^{-2}\| \tilde{\chi}_{\varepsilon,n}^+w\|_2^2\Big)\,.
\end{equation}
By (\ref{eq:95}a) we have
\begin{equation}
\label{eq:97}
  \|\nabla(\tilde{\chi}_{\varepsilon,n}^+w)\|_2^2 \leq C \Big(
  \|w\nabla\tilde{\chi}_{\varepsilon,n}^+\|_2^2+ \varepsilon^{-4}\|\tilde{\chi}_{\varepsilon,n}^+f\|_2^2+\varepsilon^{-2}\| \tilde{\chi}_{\varepsilon,n}^+w\|_2^2\Big)\,.
\end{equation}
Substituting the above into \eqref{eq:96} then yields 
\begin{displaymath}
   \| \tilde{\chi}_{\varepsilon,n}^+w\|_2^2\leq C\varepsilon^{3-\rho}
  \Big(\|w\nabla\tilde{\chi}_{\varepsilon,n}^+\|_2^2+ \varepsilon^{-4}\|\tilde{\chi}_{\varepsilon,n}^+f\|_2^2+\varepsilon^{-2}\| \tilde{\chi}_{\varepsilon,n}^+w\|_2^2\Big)\,,
\end{displaymath}
from which we easily obtain, for sufficiently small $\varepsilon$,
\begin{equation}
\label{eq:98}
   \| \tilde{\chi}_{\varepsilon,n}^+w\|_2^2\leq C\varepsilon^{3-\rho}
  \Big(\|w\nabla\tilde{\chi}_{\varepsilon,n}^+\|_2^2+ \varepsilon^{-4}\|\tilde{\chi}_{\varepsilon,n}^+f\|_2^2\Big)\,.
\end{equation}
By \eqref{eq:90} we have that for sufficiently small $\gamma_2$ and $\varepsilon$,
\begin{equation}
\label{eq:99}
  \|\tilde{\chi}_{\varepsilon,n}^+f\|_2 \leq  Ce^{-\gamma_2\varepsilon^{-\frac{3}{2}(1-\rho)}}\|f\|_2\,.
\end{equation}
Furthermore,  by \eqref{eq:91} and \eqref{eq:92} we have that
\begin{displaymath}
  \|w\nabla\tilde{\chi}_{\varepsilon,n}^+\|_2 \leq
  \frac{C}{\varepsilon^\rho}\|\tilde{\chi}_{\varepsilon,n-1}^+w\|_2 \,.
\end{displaymath}
Combining the above, \eqref{eq:99}, and \eqref{eq:98} then yields
\begin{displaymath}
   \| \tilde{\chi}_{\varepsilon,n}^+w\|_2\leq C\big( \varepsilon^\rho\|\tilde{\chi}_{\varepsilon,n-1}^+w\|_2+e^{-\gamma_2\varepsilon^{-\frac{3}{2}(1-\rho)}}\|f\|_2 \big)\,.
\end{displaymath}
Similarly we obtain that
\begin{displaymath}
   \| \tilde{\chi}_{\varepsilon,n}^-w\|_2\leq C\big( \varepsilon^\rho\|\tilde{\chi}_{\varepsilon,n-1}^+w\|_2+e^{-\gamma_2\varepsilon^{-\frac{3}{2}(1-\rho)}}\|f\|_2 \big)\,.
\end{displaymath}
The above pair of inequalities, when recursively applied, readily
yield (\ref{eq:94}a).  

We begin the proof of (\ref{eq:94}b) by combining \eqref{eq:97} and
(\ref{eq:94}a) to obtain
\begin{equation}
\label{eq:100}
  \|\nabla(\tilde{\chi}_{\varepsilon,n}^+w)\|_2 \leq C_n(\varepsilon^{n\rho-1}\|w\|_2+e^{-\gamma_2\varepsilon^{-\frac{3}{2}(1-\rho)}}\|f\|_2) \,.
\end{equation}
Furthermore, we have that
\begin{multline*}
  \|\tilde{\chi}_{\varepsilon,n}^+\Delta w\|_2 \leq
  \frac{C}{\varepsilon^3}\|(V-V(x_0))\tilde{\chi}_{\varepsilon,n}^+w\|_2\\ +\frac{C}{\varepsilon^2}\|\tilde{\chi}_{\varepsilon,n}^+w\|_2 +
  \frac{C}{\varepsilon^3}\|\tilde{\chi}_{\varepsilon,n}^+f\|_2 \leq C_n(\varepsilon^{n\rho-3}\|w\|_2+e^{-\gamma_2\varepsilon^{-\frac{3}{2}(1-\rho)}}\|f\|_2) \,.
\end{multline*}
As,
\begin{displaymath}
  \|\Delta(\tilde{\chi}_{\varepsilon,n}^+w)\|_2\leq \frac{C}{\varepsilon^\rho}
  \|\nabla(\tilde{\chi}_{\varepsilon,n-1}^+w)\|_2 +
  \frac{C}{\varepsilon^{2\rho}}\|\tilde{\chi}_{\varepsilon,n-1}^+w\|_2 + \|\tilde{\chi}_{\varepsilon,n}^+\Delta w\|_2 \,,
\end{displaymath}
we readily conclude that
\begin{displaymath}
  \|\Delta(\tilde{\chi}_{\varepsilon,n}^+w)\|_2 \leq C_n(\varepsilon^{n\rho-3}\|w\|_2+e^{-\gamma_2\varepsilon^{-\frac{3}{2}(1-\rho)}}\|f\|_2) \,.
\end{displaymath}
Standard elliptic estimates, together with \eqref{eq:100} then yield
 (\ref{eq:94}b), after repeating the same argument for $\tilde{\chi}_{\varepsilon,n}^-w$.
 \end{proof}

Before we attempt to estimate $(\A_h-\lambda^*)^{-1}f$ we need yet the
following auxiliary estimate.
\begin{lemma}
  \label{lem:aux1}
Under the same conditions of Lemma \ref{lem:Gamma} we have that
\begin{subequations}
\label{eq:101}
\begin{empheq}[left={\empheqlbrace}]{alignat=2}
& \|\nabla w\|_2 \leq \frac{C}{\varepsilon}\|w\|_2+ \frac{C}{\varepsilon^2}\|f\|_2 \,,& &\\
& \|D^2w\|_2 \leq \frac{C}{\varepsilon^{3-\rho}}\|w\|_2+ \frac{C}{\varepsilon^3}\|f\|_2\,, & &
\end{empheq}
\end{subequations}
where $w=(\A_h-\lambda^*)^{-1}f$ and $0<\rho<1$.
\end{lemma}
\begin{proof}
  As
  \begin{displaymath}
    \|\nabla w\|_2^2 = \frac{\alpha}{\varepsilon^3c^4}(\lambda^*\|w\|_2^2 + \Re\langle w,f\rangle)\,,
  \end{displaymath}
we readily obtain (\ref{eq:101}a). To prove (\ref{eq:101}b) we first
note that
\begin{equation}
\label{eq:102}
  \|\Delta w\|_2 \leq \frac{C}{\varepsilon^3}(\|(V-V(x_0))w\|_2 + \lambda^*\|w\|_2 + \|f\|_2)
\end{equation}
Let 
\begin{displaymath}
\zeta^2=1-(\tilde{\chi}_{\varepsilon,n}^-)^2-(\tilde{\chi}_{\varepsilon,n}^+)^2\,.
\end{displaymath}
By \eqref{eq:94} we have, for sufficiently large $n$,
\begin{multline*}
 \|(V-V(x_0))w\|_2 \leq C(\|\tilde\chi_{\varepsilon,n}^-w\|_2+\|\tilde\chi_{\varepsilon,n}^+w\|_2)+\|\zeta(V-V(x_0))w\|_2 \\
 \leq  C(\varepsilon^{n\rho-1}\|w\|_2+e^{-\gamma_2\varepsilon^{-\frac{3}{2}(1-\rho)}}\|f\|_2 + \varepsilon^\rho\|w\|_2) \leq C(\varepsilon^\rho\|w\|_2+e^{-\gamma_2\varepsilon^{-\frac{3}{2}(1-\rho)}}\|f\|_2)\,,
\end{multline*}
which, when substituted into \eqref{eq:102}, yields (\ref{eq:101}) with the aid of
standard elliptic estimates.
\end{proof}

Lemmas \ref{lem:3.2} and \ref{lem:Gamma} can now be used to estimate
$(\A_h-\lambda^*)^{-1}f$ in the close vicinity of $x_0$ where $\lambda^*\in\partial
B(\Lambda_0,(c^2r\varepsilon^{3/2}/\alpha))$, $r\in(0,1)$ being chosen so that
 $\partial B(\Lambda_0,(c^2r\varepsilon^{3/2}/\alpha))\subset\rho(\A_h)$, where
\begin{equation}
\label{eq:140}
   \Lambda_0= \frac{\varepsilon c^2}{\alpha} (\lambda_0+\varepsilon^{1/2}\lambda_2) \,.
 \end{equation}

\begin{lemma}
  Let $f\in L^\infty(\Omega,\C)$ satisfy \eqref{eq:90}, and $7/8<\rho<1$. Let
  $w=(\A_h-\lambda^*)^{-1}f^*$  and
  $\zeta_0$ be given by
  \begin{equation}
\label{eq:103}
      \zeta_0^*(\varepsilon,\rho)=[1- (\tilde{\chi}_{\varepsilon,n}^-)^2
      -(\tilde{\chi}_{\varepsilon,n}^+)^2]{\mathbf 1}_{B(x_0,\delta)\cap\Omega}\,,
  \end{equation}
where $\delta>0$ is so chosen so that $B(x_0,\delta)\cap\Gamma=\{x_0\}$.
Then, 
   \begin{equation}
    \label{eq:104}
\|\zeta_0^*w^*\|_2\leq \frac{C}{r}(\varepsilon^{-3/2}\|f\|_2+\varepsilon^{1/8}\|w^*\|_2)\,.
  \end{equation}
\end{lemma}
\begin{proof}
Clearly,
\begin{displaymath}
  (\A_h-\lambda^*)(\zeta_0^*w^*)= \zeta_0^*f^* +[\A_h,\zeta_0^*]w^*
\end{displaymath}
  We next write
  \begin{displaymath}
    \A_h = \A_0  + \Dg^* \,,
  \end{displaymath}
where $A_0$ is given by 
\begin{displaymath}
    \A_0 = -\frac{\varepsilon^3c^4}{\alpha^3}(\partial_{tt}+\partial_{ss}) + i(ct+\alpha s^2)\,,
 \end{displaymath}
and
\begin{displaymath}
  \Dg^* = -\frac{\varepsilon^3c^4}{\alpha^3}\Upsilon + i(V-V(x_0)-ct-\frac{1}{2} \alpha
  s^2)\,,
\end{displaymath}
where $\Upsilon$ is given by \eqref{eq:37}.
Then,
\begin{displaymath}
(\mathcal{A}_0-\lambda^*)(\zeta_0^*w^*)=\zeta_0f^*-\Dg^* (\zeta_0^*w^*)+[\A_h,\zeta_0^*] w^*\,.
\end{displaymath}
Applying the transformation \eqref{eq:39} yields
\begin{equation}
  \label{eq:105}
(\B_\varepsilon-\lambda) (\zeta_0w)=\frac{\alpha}{\varepsilon c^2}\zeta_0f+[\B_\varepsilon,\zeta_0] w-R(\zeta_0w)\,.
\end{equation}
where $f$, $\zeta_0$, and $w$ are respectively obtained from
$f^*$, $\zeta_0^*$, and $w^*$ via the dilation
$\cdot(\xi,\tau)=\cdot^*(s,t)$, in which $(\xi,\tau)$ are given by \eqref{eq:39},
$R$ is given by \eqref{eq:R} and $\lambda = \frac{\alpha}{\varepsilon c^2}\lambda^*$. 

We next apply to \eqref{eq:105} the operator $P_1$ defined in
\eqref{eq:80}. Since $\B_\varepsilon$ and $P_1$ commute, we easily obtain from
\eqref{eq:85} that
\begin{equation}
 \label{eq:106}                
  \|P_1(\zeta_0w)\|_2 \leq C(\varepsilon^{-1}\|f\|_2+ \|[\B_\varepsilon,\zeta_0]w\|_2+\|R(\zeta_0w)\|_2) \,.
\end{equation}
We now attempt to estimate $\|R(\zeta_0w)\|_2$.  We
first note that $R$ is given by \eqref{eq:R}. We then observe that
\begin{equation}
\label{eq:107}
  \Big|\frac{\alpha}{\varepsilon c^2}[V-V(x_0)]-\tau-\varepsilon^{1/2}\frac{1}{2}\xi^2\Big|\leq
  C(\varepsilon^{5/4}\xi^3+\varepsilon^{3/4}\tau\xi+\varepsilon\tau^2) \quad \forall x\in B(x_0,\delta)\,,
\end{equation}
Since 
\begin{displaymath}
   \frac{1}{2}\Big(\tau + \frac{\varepsilon^{1/2}}{2}\xi^2\Big)\leq \frac{\alpha}{\varepsilon c^2}|V(x)-V(x_0)|\leq 2\varepsilon^{-(1-\rho)}
   \quad \forall x\in supp(\zeta_0) \,,
\end{displaymath}
we obtain that for some $C>0$
\begin{equation}
\label{eq:108}
  {\rm supp}\,\zeta_0 \subset \{(\xi,\tau) \,|\, |\xi|\leq C\varepsilon^{-3/4+\rho/2}\,, \; 0\leq
  \tau<C\varepsilon^{-(1-\rho)}\}\,.
\end{equation}
Consequently, by \eqref{eq:107} we have that
\begin{displaymath}
 \zeta_0\Big|\frac{\alpha}{\varepsilon c^2}[V-V(x_0)]-\tau-\varepsilon^{1/2}\frac{1}{2}\xi^2\Big|\leq C\varepsilon^{\frac{3\rho}{2}-1}\,.
\end{displaymath}
Hence,
\begin{equation}
\label{eq:109}
  \Big\|\Big(\frac{\alpha}{\varepsilon c^2}[V-V(x_0)]-\tau-\varepsilon^{1/2}\frac{1}{2}\xi^2\Big)\zeta_0w\Big\|_2\leq C\varepsilon^{\frac{3\rho}{2}-1}\|\zeta_0w\|_2 \,. 
\end{equation}       

To complete the estimation of $R(\zeta_0w)$, it is
necessary to bound
\begin{equation}
\label{eq:110}
  \tilde{R}(\zeta_0w)= \varepsilon^{3/2}\Big\|\tau(\zeta_0w)_{\xi\xi}\Big\|_2+
  \varepsilon^{9/4}\|\tau(\zeta_0w)_\xi\|_2+\varepsilon\|(\zeta_0w)_\tau\|_2\,.
\end{equation}
Since by \eqref{eq:108} we have that
\begin{displaymath}
  \|\zeta_0\|_{C^{2,0}} \leq C \,,
\end{displaymath}
we have by \eqref{eq:39}, (\ref{eq:101}), and \eqref{eq:108} that
\begin{equation}
  \label{eq:111}
\Big\|\tau(\zeta_0w)_{\xi\xi}\Big\|_2 \leq C\Big( \frac{1}{\varepsilon^{3/2-\rho}}\|w\|_2
+\frac{1}{\varepsilon^{5/2-\rho}}\|f\|_2\Big) \,.
\end{equation}
Furthermore,
\begin{equation}
\label{eq:112}
  \|\tau(\zeta_0w)_\xi\|_2\leq C\Big( \frac{1}{\varepsilon^{1/4}}\|w\|_2
+\frac{1}{\varepsilon^{9/4-\rho}}\|f\|_2\Big) \,,
\end{equation}
and
\begin{displaymath}
  \|(\zeta_0w)_\tau\|_2\leq C( \|w\|_2
+\varepsilon^{-1}\|f\|_2) \,.
\end{displaymath}
Substituting the above together with \eqref{eq:112} and \eqref{eq:111}
into \eqref{eq:110} then yields
\begin{equation}
\label{eq:113}
  \tilde{R}(\zeta_0w)\leq  C(\varepsilon^\rho\|w\|_2 + \|f\|_2)\,.
\end{equation}
Combining the above with \eqref{eq:109} yields
\begin{equation}
\label{eq:114}
 \|R(\zeta_0w)\|_2\leq  C(\varepsilon^{\frac{3\rho}{2}-1}\|w\|_2 + \|f\|_2)\,.
\end{equation}

We now turn to estimate $[\B_\varepsilon,\zeta_0] w$.  From
\eqref{eq:94} we learn that, for any $n\in\N$, there
exists some $\varepsilon_0(n)$, such that for all $\varepsilon<\varepsilon_0(n)$ we have 
 \begin{multline}
\label{eq:115}
   \|[\B_\varepsilon,\zeta_0] w\|_2= \frac{\alpha}{c}\varepsilon^{-7/8}\Big\| \frac{\alpha}{\varepsilon c^2} [\A_h,\zeta_0^*]w^*\Big\|_2\leq \\C\varepsilon^{9/8}[\varepsilon^{-2\rho}(\|
   \tilde{\chi}_{\varepsilon,n-1}^-w^*\|_2 +  \| \tilde{\chi}_{\varepsilon,n-1}^+w^*\|_2)+ \varepsilon^{-\tilde{\rho}}(\|
   \nabla(\tilde{\chi}_{\varepsilon,n-1}^-w^*)\|_2\\ +\| \nabla(\tilde{\chi}_{\varepsilon,n-1}^+w^*)\|_2)]\leq
   C_n(\varepsilon^{n\rho-15/8}\|w^*\|_2+
   e^{-\gamma_2\varepsilon^{-\frac{3}{2}(1-\rho)}}\|f^*\|_2)\\\leq
   C_n(\varepsilon^{n\rho-1}\|w\|_2+ e^{-\gamma_2\varepsilon^{-\frac{3}{2}(1-\rho)}}\|f\|_2)\,. 
 \end{multline}
Substituting the above together with \eqref{eq:114} into
\eqref{eq:106} yields
\begin{equation}
  \label{eq:116}
 \|P_1(\zeta_0w)\|_2 \leq C(\varepsilon^{\frac{3\rho}{2}-1}\|w\|_2+ \|f\|_2) \,.
\end{equation}

We now turn to estimate $\Pi_0(w)$.  Taking the inner product of
\eqref{eq:105} in $L^2(\R_+,\C)$ with $\bar{v}_0$ yields
\begin{equation}
\label{eq:117}
  (\LL_\xi -\tilde{\lambda})w_0 =
  \varepsilon^{-1/2}\Big\langle\bar{v}_0,\frac{\alpha}{\varepsilon c^2}\zeta_0f-R(\zeta_0w)+[\B_\varepsilon,\zeta_0]
  w\Big\rangle_{\R_+}\,, 
\end{equation}
where $w_0=\langle\bar{v}_0,\zeta_0w\rangle$, and $\tilde{\lambda}=\varepsilon^{-1/2}(\lambda-\lambda_0)$.
(Note that $\Pi_0(\zeta_0w)=w_0(\xi)v_0(\tau)$.)  Multiplying \eqref{eq:117}
by $\bar{w}_0$ and integrating by parts yields, from the imaginary
part
\begin{displaymath}
  \|\xi w_0\|_{L^2(\R)}^2\leq C\big(\|w_0\|_{L^2(\R)}^2 + \varepsilon^{-1/2}
  |\langle\bar{v}_0w_0,\varepsilon^{-1}\zeta_0f-R(\zeta_0w)+[\B_\varepsilon,\zeta_0]w\rangle |\big)\,. 
\end{displaymath}
We now use \eqref{eq:113}, \eqref{eq:115}, and \eqref{eq:107} to obtain
that
\begin{multline*}
  \|\xi w_0\|_{L^2(\R)}\leq C(\|w_0\|_{L^2(\R)} + \varepsilon^{-3/2}\|f\|_2 +\varepsilon^{\rho-1/2}\|w\|_2 +\\
  \varepsilon^{3/4}\|\xi^3\zeta_0w\|_2 + \varepsilon^{1/4}\|\tau\xi\zeta_0w\|_2+\varepsilon^{1/2}\|\tau^2\zeta_0w\|_2)
\end{multline*}
In view of \eqref{eq:108} we then have
\begin{equation}
\label{eq:118}
   \|\xi w_0\|_{L^2(\R)}\leq C(\|w_0\|_{L^2(\R)} + \varepsilon^{-3/2}\|f\|_2 +\varepsilon^{\rho-1/2}\|w\|_2 +
  \varepsilon^{1/4}\|\xi\zeta_0w\|_2)\,.
\end{equation}
We now use \eqref{eq:116} to obtain
 \begin{displaymath}
  \|\xi\zeta_0w\|_2 \leq \|\xi P_1(\zeta_0w)\|_2 + \|\xi w_0\|_{L^2(\R)} \leq
  C(\varepsilon^{2\rho-7/4}\|w\|_2+\varepsilon^{-3/2}\|f\|_2)+ \|\xi w_0\|_{L^2(\R)}  \,.
\end{displaymath}
Substituting the above into \eqref{eq:118} then yields
\begin{displaymath}
  \|\xi w_0\|_{L^2(\R)}\leq C(\|w_0\|_{L^2(\R)} +
  \varepsilon^{2\rho-\frac{3}{2}}\|w\|_2+\varepsilon^{-3/2}\|f\|_2)\,,
\end{displaymath}
and hence,
\begin{displaymath}
  \|\xi\zeta_0w\|_2 \leq C(\|w_0\|_{L^2(\R)} +
  \varepsilon^{2\rho-\frac{7}{4}}\|w\|_2+\varepsilon^{-3/2}\|f\|_2)\,.
\end{displaymath}
From the above and \eqref{eq:116} once again we can conclude that
\begin{equation}
\label{eq:119}
   \|\xi^3\zeta_0w\|_2 \leq C\varepsilon^{-3/2+\rho}\|\xi\zeta_0w\|_2 \leq
   C\varepsilon^{-3/2+\rho}(\|w_0\|_{L^2(\R)} +
  \varepsilon^{2\rho-\frac{7}{4}}\|w\|_2+\varepsilon^{-3/2}\|f\|_2)\,.
\end{equation}
Similarly, we obtain
\begin{displaymath}
  \|\xi\tau\zeta_0w\|_2 \leq C\varepsilon^{-(1-\rho)}(\|w_0\|_{L^2(\R)} +  \varepsilon^{2\rho-\frac{7}{4}}\|w\|_2+\varepsilon^{3/2}\|f\|_2)\,.
\end{displaymath}
The above, together with \eqref{eq:119}, \eqref{eq:107}, and
\eqref{eq:108} yield the following improvement of \eqref{eq:109}
(recall that $\|\Pi_0(w)\|_2\leq C\|w\|_2$)
\begin{displaymath}
  \Big\|\big[\frac{\alpha}{\varepsilon
    c^2}[V-V(x_0)]-\tau-\varepsilon^{1/2}\frac{1}{2}\xi^2\big]\zeta_0w\Big\|_2\leq
  C\varepsilon^{\rho-1/4}(\|w\|_2+\varepsilon^{-3/2}\|f\|_2) \,.
\end{displaymath}
We now combine the above inequality with \eqref{eq:113} to obtain an
improved version of \eqref{eq:114}
\begin{equation}
\label{eq:120}
  \|R(\zeta_0w)\|_2\leq  C\varepsilon^{\rho-1/4}(\varepsilon^{\tilde{\rho}}\|w\|_2+\varepsilon^{3/2}\|f\|_2) \,.
\end{equation}

Returning to \eqref{eq:105} we obtain from \eqref{eq:73} that
\begin{displaymath}
   \|\zeta_0w\|_2\leq \frac{C}{r\varepsilon^{1/2}}(\varepsilon^2\|f\|_2+ \|[\B_\varepsilon,\zeta_0]w\|_2+\|R(\zeta_0w)\|_2) \,.
\end{displaymath}
With the aid of \eqref{eq:115} and \eqref{eq:120} we then obtain
\begin{displaymath}
  \|\zeta_0w\|_2\leq \frac{C}{r\varepsilon^{1/2}}(\varepsilon^{-1}\|f\|_2+ \varepsilon^{5/8}\|w\|_2) \,,
\end{displaymath}
from which \eqref{eq:104} easily follows.
\end{proof}

\begin{remark}
  Clearly, \eqref{eq:104} can be extended to the neighborhood of each
  point in $\Sg$. Thus, if we set for any $x_j\in\Sg$
  \begin{equation}
\label{eq:121}
      \zeta_j^*(\varepsilon,\rho)=[1- (\tilde{\chi}_{\varepsilon,n}^-)^2
      -(\tilde{\chi}_{\varepsilon,n}^+)^2]{\mathbf 1}_{B(x_j,\delta)\cap\Omega}\,,
  \end{equation}
where $\delta>0$ is so chosen so that $B(x_j,\delta)\cap\Gamma=\{x_j\}$ for all $j\in J_\Sg$.
Then, 
  \begin{equation}
 \label{eq:122}
\|\zeta_j^*w^*\|_2\leq \frac{C}{r}(\varepsilon^{3/2}\|f\|_2+\varepsilon^{1/8}\|w^*\|_2)\,.
  \end{equation}\end{remark}

We can now estimate $\|(\A_h-\lambda^*)^{-1}f\|$ in the simplest possible
case where $\Gamma=\{x_0\}$.
\begin{corollary}
  Let $f\in L^\infty(\Omega,\C)$ satisfy \eqref{eq:90}. Let $\lambda^*\in\partial
  B(\Lambda_0,r\varepsilon^{-1/2})\subset\rho(\A_h)$, where $\Lambda_0$ is
  given by \eqref{eq:140}, for some $\varepsilon^{1/8}\ll r<1$.
  Then, there exists $C>0$ such that for sufficiently small $\varepsilon$ we
  have
  \begin{equation}
    \label{eq:123}
\|(\A_h-\lambda^*)^{-1}f\|_2 \leq \frac{C}{\varepsilon^{3/2}r}\|f\|_2\,.
  \end{equation}
\end{corollary}
\begin{proof}
  Since $\Gamma=\{x_0\}$ we may set with any loss of
  generality $\Omega=\Omega_+$. Hence, we have that
  $\chi_{\varepsilon,n}^+=\zeta_0^*$, where $\zeta_0^*$ is defined by \eqref{eq:103}. Let $w=(\A_h-\lambda)^{-1}f$. Then,
  \begin{displaymath}
    \|w\|_2^2 = \|\chi_{\varepsilon,n}^+w\|_2^2+ \|\tilde{\chi}_{\varepsilon,n}^+w\|_2^2=
    \|\zeta_0^*w\|_2^2+ \|\tilde{\chi}_{\varepsilon,n}^+w\|_2^2 \,.
  \end{displaymath}
The corollary now easily follows from (\ref{eq:94}a) and
\eqref{eq:104}. 
\end{proof}

Consider next the general case where $\Gamma\setminus\{x_0\}\neq\emptyset$. We begin by
defining some local approximations of the operator $\tilde\A_h$.  Let
$\rho\in(7/8,1)$, and then define two sets of indices $J_{\pa\Omega} =
J_{\pa\Omega}(\varepsilon)$ and $J_\Omega = J_\Omega(\varepsilon)$. Set then $J=J_{\pa\Omega}\cup J_\Omega$
and let $\delta>0$ be the same as in \eqref{eq:103}. Next, choose a
sequence of points $(x_j)_{j\in J} = (x_j(\varepsilon))_{j\in
  J}\subset\bar{\Omega}\setminus \Union_{j\in J_\Sg}B(x_j,\delta)$, where $x_j\in\pa\Omega$ (respectively $x_j\in\Omega$) if
$j\in J_{\pa\Omega}$ (respectively $j\in\Omega$), such that 
\begin{displaymath}
 \bar\Omega\setminus\Union_{j\in J_\Sg}B(x_j,\delta)\subset\Union_{j\in J}B(x_j,\varepsilon^\rho)\,.
\end{displaymath}
Let $(\eta_j)_{j\in J}$ be a family of cutoff functions associated with the partition above, namely
$\eta_j(x) = 1$ if $x\in B(x_j,\varepsilon^\rho/2)$, ${\supp}\eta_j\subset B(x_j,\varepsilon^\rho)$, and 
\begin{displaymath}
 \forall x\in\bar\Omega\setminus\Union_{j\in J_\Sg}B(x_j,\delta)\,,~~~\sum_{j\in J}\eta_j(x)^2 = 1\,.
\end{displaymath}
We further assume that for all $j\in J$, $\|\nabla\eta_j\|_\infty =
\mathcal{O}(\varepsilon^{-\rho})$ and $\|\Delta\eta_j\|_\infty = \mathcal{O}(\varepsilon^{-2\rho})$. 
Finally we set, for all $j\in J$,
\begin{displaymath}
 \chi_j = \eta_j\mathbf{1}_{\bar\Omega}\,.
\end{displaymath}
In the neighborhood of each point $x_j$, $j\in J_\Omega$, we shall
approximate $\A_h$ by the following operator:
\begin{subequations}
\label{eq:124}
  \begin{equation}
 \A_{j,h} := -\frac{\varepsilon^3c^4}{\alpha^3}\Delta+i({\bf c}_j. x+V(x_j)-V(x_0))\,,~~~~{\bf c}_j = (c_j^1,c_j^2) = \nabla V(x_j)\,,
\end{equation}
whose domain is given by 
\begin{equation}
 D(\A_{j,h}) = H^2(\mathbb{R}^2 ; \mathbb{C})\cap L^2(\mathbb{R}^2, |x|^2dx ; \mathbb{C})\,.
\end{equation}
\end{subequations}
In the neighborhood of the boundary points $x_j$, $j\in J_{\pa\Omega}$, we
use different approximate operators, depending on the local
behaviour of $V$. To this end, denote by $J_{\pa\Omega}^1\subset J_{\pa\Omega}$
the set of indices $j$ such that $x_j\in\partial\Omega_\perp$ and
\begin{displaymath}
|\nabla V(x_j)| = |\nabla V(x_0)| = \min_{x\in\pa\Omega_\perp}|\nabla V(x)|\,.
\end{displaymath}
Notice that $J_{\pa\Omega}^1$ may be an empty set, since
$x_0\not\in\bar{\Omega}\setminus B(x_0,\delta)$. We then let $J_{\pa\Omega}^2 = J_{\pa\Omega}\setminus
J_{\pa\Omega}^1\,$ and $J_{\pa\Omega}^3=J_{\pa\Omega}^1\setminus J_\Sg$.  In the
neighborhood of the boundary points $x_j$ for $j\in J_{\pa\Omega}^2$, we
use the following approximation of $\A_h$. Let $(t,s)$ be the same
curvilinear coordinate system as defined in Section \ref{sec:2},
centered at $x_j$. In these coordinates the leading order
approximation of $\A_h$ reads
\begin{subequations}
\label{eq:125}
  \begin{equation}
 \A_{j,h} = -\frac{\varepsilon^3c^4}{\alpha^3}\Delta + i({\bf c}_j.(t,s)+V(x_j)-V(x_0))\,,~~~~{\bf c}_j = (c_j^1,c_j^2) = \nabla V(x_j)\,,
\end{equation}
with the following domain
\begin{equation}
 D(\A_{j,h}) = H_0^1(\mathbb{R}_+^2 ; \mathbb{C})\cap H^2(\mathbb{R}_+^2 ;
 \mathbb{C})\cap L^2(\mathbb{R}_+^2, (t^2+s^2)dtds ; \mathbb{C})\,. 
\end{equation}
\end{subequations}

In the following we provide resolvent estimates on the approximate operators $\A_{j,h}$
introduced above. These estimates are stated in the following lemma
\begin{lemma}
\label{lem:ResModels}
There exists $r_0>0$ such that, for all $r\in(0,r_0)$ and $j\in J$,
$\pa B(\Lambda_0,r\varepsilon^{-1/2})\subset\rho(\A_{j,h})$, where $\Lambda_0$
is given by \eqref{eq:140}.
Moreover, there exists $C>0$ such that for all $\lambda^*\in\pa
B(\Lambda_0,r\varepsilon^{-1/2})$ and 
 for all $j\in J_{\Omega}\cup J_{\pa\Omega}^2$,
\begin{equation}
\label{eq:ResModels1}
 \|(\A_{j,h}-\lambda^*)^{-1}\|_2 \leq \frac{C}{\varepsilon}\,.
\end{equation}
\end{lemma}
\begin{proof}
  Let $j\in J_\Omega$. Recall that the operator $\A_{j,h}$ is given in this
  case by \eqref{eq:125}. It has been established in 
  \cite{al08,he09b} that $\A_{j,h}$ has empty spectrum, and for all
  $\omega\in\mathbb{R}$ there exists $C_\omega>0$ such that
\begin{equation}\label{eq:126}
 \sup_{\Re z\leq\omega}\big\|(-\Delta+i{\bf c}_j.x-z)^{-1}\big\| \leq C_\omega\,.
\end{equation}
Since the scale change $x\mapsto\alpha/(\varepsilon c^{4/3})x$ gives
\begin{equation}\label{eq:127}
 \|(\A_{j,h}-\lambda^*)^{-1}\| =
 \frac{\alpha}{\varepsilon c^{4/3}}\left\|\Big(-\Delta+i\Big[\frac{\alpha}{\varepsilon c^{4/3}}\big(V(x_j)-V(x_0)\big)+{\bf
     c}_j. x\Big]-\frac{\alpha}{\varepsilon c^{4/3}}\lambda^*\Big)^{-1}\right\|\,. 
\end{equation}
and since $\alpha/(\varepsilon c^{4/3})\lambda^*$ remains bounded as $\varepsilon\to0$,
\eqref{eq:126} and \eqref{eq:127} easily yield \eqref{eq:ResModels1}
for any $j\in J_\Omega$.

The same argument can be used in the case where $j\in J_{\pa\Omega}^2$ with
$x_j\notin\pa\Omega_\perp$, since the operator $-\Delta+ic_j^1t+ic_j^2s$ on
$\mathbb{R}_+^2$ has empty spectrum and satisfies \eqref{eq:126} as
well as soon as $c_j^2\neq0$, see Theorem \ref{thm:ap1}.

We next consider the case where $j\in J_{\pa\Omega}^2$ and $x_j\in\pa\Omega_\perp$.
Then,
\begin{displaymath}
 \A_{j,h} = -\frac{\varepsilon^3c^4}{\alpha^3}\Delta + i\big(c_jt+V(x_j)-V(x_0)\big)
\end{displaymath}
where $c_j:=c_j^1$. The domain $D(\A_{j,h})$ is given by
(\eqref{eq:125}b).
Suppose that $c_j>0$ (otherwise apply the same argument to the operator $\A_{j,h}^*$).
Denote by $\A_0^\perp$ the Dirichlet realization on $\mathbb{R}_+^2$ of the operator $-\Delta+it$. Then,
the scale change
\begin{displaymath}
 (t,s)\longmapsto \frac{\alpha c_j^{1/3}}{\varepsilon c^{4/3}}\,(t,s)
 \end{displaymath}
gives
\begin{equation}\label{eq:128}
  \|(\A_{j,h}-\lambda^*)^{-1}\| =
  \frac{\alpha c_j^{1/3}}{\varepsilon c^{4/3}}\left\|\Big(\A_0^\perp+ i\frac{\alpha
      c_j^{1/3}}{\varepsilon c^{4/3}}\big(V(x_j)-V(x_0)\big) - \frac{\alpha c_j^{1/3}}{\varepsilon c^{4/3}}\lambda^*\Big)^{-1}\right\|\,.
\end{equation}
By the definition of $J_{\pa\Omega}^2$, we have $c_j<c$. Hence for any fixed $\delta_0\in(0,1)$ we have
\begin{displaymath}
  \frac{\alpha c_j^{1/3}}{\varepsilon c^{4/3}}\lambda^* = \left(\frac{c}{c_j}\right)^{2/3}\lambda_0 + \mathcal{O}(\varepsilon^{1/2}) \leq (1-\delta_0)\lambda_0
\end{displaymath}
for all sufficiently small $\varepsilon$. It has been established
in \cite{he09b} that
\begin{displaymath}
 \sup_{\Re z\leq(1-\delta_0)\lambda_0}\|(\A_0^\perp-z)^{-1}\| < +\infty\,.
\end{displaymath}
Consequently, \eqref{eq:ResModels1} follows
from \eqref{eq:128} and the above estimate. 
\end{proof}

We now extend \eqref{eq:123} to the general case
\begin{proposition}
  \label{lem:general}
  Let $\varepsilon^{1/8}\ll r<1$. Under the assumptions of Theorem \ref{thm:1},
  \eqref{eq:123} holds for any $f\in L^\infty(\Omega,\C)$ satisfying
  \eqref{eq:90}, and $\lambda^*\in\pa
  B(\Lambda_0,r\varepsilon^{-1/2})$.
\end{proposition}
\begin{proof}
  Let $w=(\A_h-\lambda^*)^{-1}f$. Let $j\in J_{\partial\Omega}^2\cup J_\Omega$. Clearly
  \begin{equation}
\label{eq:129}
    (\A_{j,h}-\lambda^*)(\chi_jw)= [\A_h,\chi_j]w - (\A_h-\A_{j,h})(\chi_jw) \,.
  \end{equation}
We now attempt to estimate the right-hand-side of \eqref{eq:129}. Clearly,
\begin{equation}
\label{eq:130}
  \|[\A_h,\chi_j]w\|_2\leq C\varepsilon^{-2\rho}\|w\|_{L^2(B(x_j,\varepsilon^\rho))} +
  C\varepsilon^{-\rho}\|\nabla(\chi_jw)\|_2 \,.
\end{equation}
As
\begin{displaymath}
  \Re\langle\chi_j^2w,(\A_h-\lambda^*)w\rangle=\|\nabla(\chi_jw)\|_2^2- \lambda^*\|\chi_jw\|_2^2- \|w\nabla\chi_j\|_2^2=0\,, 
\end{displaymath}
we obtain that
\begin{equation}
\label{eq:131}
  \|\nabla(\chi_jw)\|_2 \leq C\varepsilon^{-1}\|w\|_{L^2(B(x_j,\varepsilon^\rho))} \,,
\end{equation}
which, when substituted into \eqref{eq:130} yields
\begin{equation}
  \label{eq:132}
  \|[\A_h,\chi_j]w\|_2\leq C\varepsilon^{-(1+\rho)}\|w\|_{L^2(B(x_j,\varepsilon^\rho))}\,.
\end{equation}
We now attempt to estimate $(\A_h-\A_{j,h})(\chi_jw)$. 
By \eqref{eq:125} and \eqref{eq:124} we have that
\begin{displaymath}
  \A_h-\A_{j,h}= i\frac{\alpha^3}{\varepsilon^3c^4}\big(V(x)-V(x_j)-{\bf c}_j. (x-x_j)\big)\,.
\end{displaymath}
Consequently,
\begin{displaymath}
   \|(\A_h-\A_{j,h})(\chi_jw)\|_2 \leq C\varepsilon^{-3+2\rho}\|w\|_{L^2(B(x_j,\varepsilon^\rho))}\,.
\end{displaymath}
Combining the above with \eqref{eq:132}, \eqref{eq:129}, and
\eqref{eq:ResModels1} yields
\begin{equation}
  \label{eq:133}
\|\chi_jw\|_2 \leq C\varepsilon^{2\rho-1}\|w\|_{L^2(B(x_j,\varepsilon^\rho))}\,.
\end{equation}

Consider next the case where $j\in J_{\pa\Omega}^3$. Here we have
\begin{displaymath}
  \Im\langle\chi_j^2w,(\A_h-\lambda^*)w\rangle =
  \frac{\alpha^3c_j}{\varepsilon^3c^4}\||V(\cdot)-V(x_0)|^{1/2}\chi_jw\|_2^2 -
  \Im\lambda^*\|\chi_jw\|_2^2+2\Im\langle w\nabla\chi_j,\chi_j\nabla w\rangle=0\,.
\end{displaymath}
By \eqref{eq:3}, there exists $\delta_1>0$ such that
$|V(x_j)-V(x_0)|>\delta_1$. Consequently,
\begin{displaymath}
  \|\chi_jw\|_2^2 \leq C[\varepsilon\|\chi_jw\|_2^2 + \varepsilon^3\|w\nabla\chi_j\|_2\|\chi_j\nabla w\|_2]\,.
\end{displaymath}
With the aid of \eqref{eq:131}, which is valid for every $j\in J$, we
then obtain
\begin{equation}
\label{eq:134}
   \|\chi_jw\|_2 \leq C\varepsilon^{1-\rho/2}\|w\|_{L^2(B(x_j,\varepsilon^\rho))}\,.
\end{equation}
Combining \eqref{eq:134} and \eqref{eq:133} then  yields
\begin{equation}
  \label{eq:135}
\|w\|_{L^2\big(\Omega\setminus\Union_{j\in J_\Sg}B(x_j,\delta)\big)}\leq
C\varepsilon^{1-\rho/2}\sum_{j\in J_\Omega\cup J_{\partial\Omega}^2}\|w\|_{L^2(B(x_j,\varepsilon^\rho))}\leq
C\varepsilon^{1-\rho/2}\|w\|_2 \,.
\end{equation}

We conclude the proof by recalling that for all $j\in J_\Sg$ we have, by
\eqref{eq:122}
\begin{equation}
\label{eq:136}
  \|\zeta_j^*w\|_2\leq \frac{C}{r}(\varepsilon^{3/2}\|f\|_2+\varepsilon^{1/8}\|w\|_2)\,.
\end{equation}
Furthermore, let
\begin{displaymath}
  \tilde{\zeta_j^*}^2 +(\zeta_j^*)^2 =\mathbf{1}_{B(x_j,\delta)} \,.
\end{displaymath}
Then, by (\ref{eq:94}a)
\begin{displaymath}
   \|\tilde{\zeta_j^*}w\|_2^2 \leq   \| \tilde{\chi}_{\varepsilon,n}^+w\|_2^2 + \|
   \tilde{\chi}_{\varepsilon,n}^-w\|_2^2\leq C_n(\varepsilon^{n\rho-1}\|w\|_2+ e^{-c\varepsilon^{-\frac{3}{2}(1-\rho)}}\|f\|_2)\,.
\end{displaymath}
which, together with \eqref{eq:136} and \eqref{eq:135} yields \eqref{eq:90}.
\end{proof}

\begin{proof}[Proof of Theorem \ref{thm:1}]
  Let $U$ be given by \eqref{eq:49} and $\Lambda_0$ be given by
  \eqref{eq:140}.  Let $f= (\A_h-\Lambda_0)(\eta_{\varepsilon^{1/2}}U)$.
  Then, for $\lambda^*\in\partial
  B(\Lambda_0,r\varepsilon^{-1/2})\subset\rho(\A_h)$ where
  $\varepsilon^{1/8}\ll r<1$, 
  \begin{displaymath}
  (\A_h-\lambda^*)(\eta_{\varepsilon^{1/2}}U) = f + (\Lambda_0-\lambda)\eta_{\varepsilon^{1/2}}U\,.
\end{displaymath}
Hence
\begin{displaymath}
  \langle\eta_{\varepsilon^{1/2}}U, (\A_h-\lambda^*)^{-1}(\eta_{\varepsilon^{1/2}}U)\rangle =-\frac{1}{\lambda-\Lambda_0}[1-\langle\eta_{\varepsilon^{1/2}}U, (\A_h-\lambda)^{-1}f\rangle]
\end{displaymath}
By \eqref{eq:123} and \eqref{eq:52} we then obtain that
\begin{displaymath}
  \|(\A_h-\lambda)^{-1}f\|_2\leq C\frac{\varepsilon^{-3/2}}{r}\|f\|_2\leq
  C\frac{\varepsilon^{1/2}}{r}\leq C\varepsilon^{1/4}\,.
\end{displaymath}
Consequently
\begin{displaymath}
  \frac{1}{2\pi i}\oint_{\partial B(\Lambda_0,r\varepsilon^{-3/2})}\langle\eta_{\varepsilon^{1/2}}U,
  (\A_h-\lambda)^{-1}(\eta_{\varepsilon^{1/2}}U)\rangle \leq -1+C\varepsilon^{1/4} \,.
\end{displaymath}
Hence $(\A_h-\lambda)^{-1}$ is not holomorphic in $B(\Lambda_0,r\varepsilon^{-3/2})$ and
the Theorem is proved via \eqref{eq:137}.
\end{proof}

\appendix

\section{Spectral analysis of \eqref{eq:125})}

In the following we provide the spectrum, semigroup estimates, and
resolvent estimates for the operator $\A_{j,h}$ given by
(\ref{eq:125}). This operator has already been investigated in
\cite{al08,he09b}, but since resolvent estimates have not been
obtained there we derive them here.

Let $\mathbf{c}=(c^1,c^2)\in\mathbb{R}^2$ such that $c^2\neq0$.  We study
here the spectrum and the resolvent of the Dirichlet realization in
$\mathbb{R}_+^2 = \{(t,s)\in\mathbb{R}^2 : t>0\}$ of
$-\Delta+i(c^1t+c^2s)\,$, whose domain is given by (\ref{eq:125}b).  The
imaginary part of the potential
\[
 \ell(t,s) = \mathbf{c}\cdot (t,s)
\]
does not have a constant sign, hence we are unable to use the variational approach to define the operator.
We shall instead define the operator by separation of variables.

Let
\begin{equation}\label{defCA_x'}
 \A_s = -\pa_s^2+ic^2s\,,
\end{equation}
and let $\A_t^+$ be the Dirichlet realization in $\mathbb{R}_+$ of the complex Airy operator
\begin{equation}\label{defCA_xn}
 -\frac{d^2}{dt^2}+ic^1t\,.
\end{equation}
Both $\A_s$ and $\A_t^+$ are maximally accretive and hence they serve
as generators of contraction semigroups $(e^{-t\A_s})_{t>0}$ and
$(e^{-t\A_t^+})_{t>0}$ respectively. One can easily verify that the
family $(e^{-t\A_s}\otimes e^{-t\A_t^+})_{t>0}$ is a contraction semigroup
on $L^2(\mathbb{R}_+^2)\,$. Thus, we can define the desired operator
as follows:
\begin{definition}
$\A_+$ is the generator of the semigroup $(e^{-t\A_s}\otimes e^{-t\A_t^+})_{t>0}\,$.
\end{definition}

Let $ D= D(\A_s)\otimes D(\A_t^+)$ be the set of all finite linear
combinations of functions of the form $f\otimes g=f(s)g(t)$, where $f\in
D(\A_s)$ and $g\in D(\A_t^+)$. Then it is clear that $D$ satisfies the
conditions of \cite[Theorem X.$49$]{resi78}, hence $\A_+ =
\overline{{\A_+}_{|D}}\,$. Consequently, we may chacterize $D(\A_+)$
as follows:
\begin{eqnarray}
D(\A_+) &=& \{u\in L^2(\mathbb{R}_+^2) : \exists (u_j)_{j\geq1}\subset D \,,~
 u_j\underset{\tiny{j\to+\infty}}{\overset{L^2}{\longrightarrow}}u\,,\nonumber\\
 & &\quad\quad \quad\quad  (\A_+u_j)_{j\geq1}~\textrm{is a Cauchy sequence}~\}\,.\label{caracDomA+}
\end{eqnarray}

In the following lemma we give a more constructive description of $D(\A_+)$.
\begin{lemma}
We have
\begin{equation}\label{descrDomA+}
D(\A_+) = H_0^1(\mathbb{R}_+^2)\cap H^2(\mathbb{R}_+^2)\cap L^2(\mathbb{R}_+^2 ; |\ell(t,s)|^2dtds)\,,
\end{equation}
and there exists $C>0$ such that, for all $u\in D(\A_+)\,$,
\eq\label{estDomA+}
\|\Delta u\|_{L^2(\mathbb{R}_+^2)}^2 + \|\ell u\|_{L^2(\mathbb{R}_+^2)}^2 \leq \|\A_+ u\|_{L^2(\mathbb{R}_+^2)}^2+
C\|\nabla u\|_{L^2(\mathbb{R}_+^2)}\|u\|_{L^2(\mathbb{R}_+^2)}\,.
\eeq
\end{lemma}
\textbf{Proof: }
Let $u\in D(\A_+)$ and
$(u_j)_{j\geq1}\subset D$ such that $u_j\underset{\tiny{j\to+\infty}}{\overset{L^2}{\longrightarrow}}u$
and $(\A_+u_j)_{j\geq1}$ is a Cauchy sequence. Then, using the identity
\[
 \Re\sca{\A_+u_j}{u_j} = \|\nabla u_j\|_{L^2(\mathbb{R}_+^2)}^2\,,
\]
which holds for every $j\in\N$,
we obtain that $(\nabla u_j)_{j\geq1}$ is a Cauchy sequence in
$L^2(\mathbb{R}_+^2)$ and hence
\begin{equation}\label{cvNabla}
u_j\underset{\tiny{j\to+\infty}}{\overset{H^1}{\longrightarrow}} u\,,
\end{equation}
and $u\in H_0^1(\mathbb{R}_+^2)\,$.\\
To prove (\ref{estDomA+}), we write (hereafter $\|\cdot\|$ denotes the $L^2(\R^2_+,\C)$ norm)
\begin{eqnarray}
 \|\A_+u_j\|^2 & = & \sca{(-\Delta+i\ell)u_j}{(-\Delta+i\ell)u_j} \nonumber\\
 & = & \|\Delta u_j\|^2+ \|\ell u_j\|^2 + 2\Im\sca{-\Delta u_j}{\ell u_j}\,.\label{estDomInter}
\end{eqnarray}
As
\begin{eqnarray*}
\Im \sca{-\Delta u_j}{\ell u_j} & = & \Im\int_{\mathbb{R}_+^2}\nabla u_j(t,s)\cdot\overline{\nabla (\ell u_j)(t,s)}dtds\\
 & = & \Im\left(\int_{\mathbb{R}_+^2}\ell(t,s)|\nabla u_j(t,s)|^2dtds + \int_{\mathbb{R}_+^2}\nabla u_j(t,s)\cdot\overline{\nabla\ell(t,s)}\overline{u_j(t,s)}dtds\right) \\
 & = & \Im\int_{\mathbb{R}_+^2} {\mathbf c}\cdot\nabla u_j(t,s)\overline{u_j(t,s)}dtds\,,
\end{eqnarray*}
it follows that for some $C>0\,$,
\[
 |\Im \sca{-\Delta u_j}{\ell u_j}|\leq C\,\|\nabla u_j\|\, \|u_j\|\,.
\]
Thus, by (\ref{estDomInter}), (\ref{estDomA+}) holds for $u_j$ for all
$j\in\N$. Consequently, $(u_j)_{j\geq1}$ is a Cauchy sequence in
$H^2(\mathbb{R}_+^2)$ and in $L^2(\mathbb{R}_+^2 ; |\ell(t,s)|^2dtds)$.
Hence, (\ref{descrDomA+}) follows, and so does (\ref{estDomA+}) for every
$u\in D(\A_+)\,$.
\hfill $\boxminus$\\

We now obtain the spectrum of $\A_+\,$. Since $\A_s$ has an empty
spectrum (see \cite{al08,he09b}), we expect $\sigma(\A_+)$ to be empty as
well \cite{al08}. To establish this fact we employ semigroup
estimates.
\begin{theorem}\label{thm:ap1}
We have $\sigma(\A_+) = \emptyset\,$. Moreover, for every $\omega\in\mathbb{R}\,$, there exists $C_\omega>0$ such that
\begin{equation}\label{ResA+}
\sup_{\tiny{\Re z \leq \omega}}\|(\A_+-z)^{-1}\| \leq C_\omega\,.
\end{equation}
Finally, the semigroup generated by $\A_+$ satisfies
\begin{equation}\label{decSGA+}
\forall t>0,~~\|e^{-t\A_+}\|\leq e^{-t^3/12}\,.
\end{equation}
\end{theorem}

\textbf{Proof: } Recall that $e^{-t\A_+} = e^{-t\A_s}\otimes
e^{-t\A_t^+}\,$, where $\A_s$ and $\A_t^+$ are respectively defined by
(\ref{defCA_x'}) and (\ref{defCA_xn}). Recall further the following
estimates (see \cite{he09b}):
\begin{equation}\label{estSGAiry}
 \forall t>0\,,~~~ \|e^{-t\A_s}\| = e^{-t^{3}/12}\,,
\end{equation}
and for all $\omega<|\mu_1|/2\,$ ($\mu_1$ being the rightmost zero of Airy's
function), there exists $M_\omega>0$ such that
\begin{equation}\label{estSGAiry+}
 \forall t>0\,,~~~ \|e^{-t\A_t^+}\|\leq M_\omega\, e^{-\omega t}\,.
\end{equation}
Thus, (\ref{decSGA+}) follows, and the formula
\begin{equation}\label{formResSG}
 (\A_+-z)^{-1} = \int_0^{+\infty}e^{-t(\A_+-z)}dt\,,
\end{equation}
which holds \emph{a priori} for $\Re z<0\,$, can be extended to the
entire complex plane.  Hence the resolvent of $\A_+$ is an entire
function, and we must have $\sigma(\A_+) = \emptyset$ together with (\ref{ResA+})\,.
\hfill $\square$

\paragraph{\bf Acknowledgements}
The authors are grateful to Bernard Helffer for his valuable comments.
Y. Almog was partially supported by NSF grant DMS-1109030. R. Henry
acknowledges the support of the ANR project NOSEVOL. 

\bibliography{mixed1}
\end{document}